\documentclass[twocolumn,journal,twoside]{IEEEtran}
\usepackage[T1]{fontenc}
\usepackage[bookmarks=false]{hyperref}
\usepackage[cmex10]{amsmath}
\usepackage[font={small}]{caption}

\usepackage{cite}
\usepackage{subcaption}

\usepackage{graphicx}
\usepackage{pstricks}
\usepackage{color}

\usepackage{amsmath}
\usepackage{amssymb}
\usepackage{amsthm}

\usepackage{relsize}
\usepackage{bbm}
\usepackage{bm}

\makeatletter
\def\thm@space@setup{\thm@preskip=2pt
	\thm@postskip=2pt \itshape}
\makeatother

\newtheoremstyle{newstyle}      
{} 
{} 
{\mdseries} 
{} 
{\bfseries} 
{.} 
{ } 
{} 

\theoremstyle{newstyle}

\newtheorem{theorem}{Theorem}
\newtheorem{lemma}{Lemma}

\newtheorem{corollary}{Corollary}

\theoremstyle{definition}

\newtheorem{definition}{Definition}

\theoremstyle{remark}
\newtheorem{remark}{Remark}

\begin{document}
	
		\sloppy
		
		\setlength{\belowcaptionskip}{-6pt}
		\setlength{\abovedisplayskip}{1mm}
		\setlength{\belowdisplayskip}{1mm}
		\setlength{\abovecaptionskip}{1mm}
	
	\title{Characterizing the Rate-Memory Tradeoff in Cache Networks within a Factor of 2}
		\author{Qian~Yu,~\IEEEmembership{Student~Member,~IEEE,}
	Mohammad~Ali~Maddah-Ali,~\IEEEmembership{Member,~IEEE,}
        and~A.~Salman~Avestimehr,~\IEEEmembership{Senior Member,~IEEE}
			\thanks{Manuscript received February 15, 2017; revised February 07, 2018; accepted Augest 16, 2018. A shorter version of this paper was presented at ISIT, 2017 \cite{8006555}. }
			\thanks{Q.~Yu and A.S.~Avestimehr are with the Department of Electrical Engineering, University of Southern California, Los Angeles, CA, 90089, USA (e-mail:  qyu880@usc.edu; avestimehr@ee.usc.edu).}
\thanks{M. A. Maddah-Ali is with Department of Electrical Engineering, Sharif University of Technology, Tehran, 11365, Iran (e-mail: maddah\_ali@sharif.edu).}
\thanks{
Communicated by M. Neely, Associate Editor for Shannon Theory. }
\thanks{
This work is in part supported by NSF grants CCF-1408639, NETS-1419632, and ONR award N000141612189.}
\thanks{Copyright (c) 2017 IEEE. Personal use of this material is permitted.  However, permission to use this material for any other purposes must be obtained from the IEEE by sending a request to pubs-permissions@ieee.org.}
		}
	\maketitle

	\begin{abstract} 
    We consider a basic caching system, where a single server with a database of $N$ files (e.g. movies) is connected to a set of $K$ users through a shared bottleneck link. 
     Each user has a local cache memory with a size of $M$ files. The system operates in two phases: a placement phase, where each cache memory is populated up to its size from the database, and a following delivery phase, where each user requests a file from the database, and the server is responsible for delivering the requested contents. The objective is to design the two phases to minimize the load (peak or average) of  the bottleneck link. We characterize the rate-memory tradeoff of the above caching system within a factor of $2.00884$ for both the \emph{peak rate} and the \emph{average rate} (under uniform file popularity), improving state of the arts that are within a factor of $4$ and $4.7$ respectively. Moreover, in a practically important case where the number of files ($N$) is large, we exactly characterize the tradeoff for systems with no more than $5$ users, and characterize the tradeoff within a factor of $2$ otherwise. To establish these results, we develop two new converse bounds that improve over the state of the art.
	\end{abstract}

    \section{Introduction}

    
    Caching is a common strategy to mitigate heavy peak-time communication load in a distributed network, via duplicating parts of the content in memories distributed across the network during off-peak times. In other words, caching allows us to trade distributed memory in the network for communication load reduction. Characterizing this fundamental \emph{rate-memory tradeoff} is of great practical interest, and has been a research subject for several decades. 
     For single-cache networks, the \emph{rate-memory tradeoff}  has been characterized for various scenarios  in the 80s~\cite{sleator1985amortized}. However, those techniques were found insufficient to tackle the multi-cache cases.

    There has been a surge of recent results in information theory that aim at formalizing and characterizing such rate-memory tradeoff in multi-cache networks~\cite{maddah-ali12a,maddah-ali13, pedarsani13, niesen13, ji2015order, zhang15, ji14b, lim2016information, bidokhti2016noisy, zhang2015fundamental,hachem14, karamchandani14,maddah2015cache,naderializadeh2016fundamental,7996347}. 
    In particular, a basic bottleneck caching network was considered in \cite{maddah-ali12a}, where a set of $K$ users is connected to a server through a shared error-free link. In this setting, each user has a local cache of size $M$, which can be used to prefetch the contents (a library of $N$ files).
    The objective is to design the caching functions, such that in a following delivery phase, the server can serve the user demands with efficient bandwidth usage (measured by the communication rate $R$). 
    For this case, the peak rate vs. memory tradeoff (the tradeoff between maximum $R$ over all possible user demands and $M$) was formulated and characterized within a factor of $12$~\cite{maddah-ali12a}. 
    This caching framework has been extended to many scenarios, including decentralized caching \cite{maddah-ali13}, online caching \cite{pedarsani13}, caching with nonuniform demands \cite{niesen13, zhang15, ji2015order},  device-to-device caching \cite{ji14b},  caching on file selection networks \cite{lim2016information}, caching on broadcast channels \cite{bidokhti2016noisy}, caching for channels with delayed feedback with channel state information  \cite{zhang2015fundamental}, hierarchical cache networks~\cite{hachem14, karamchandani14}, and caching on interference channels~\cite{maddah2015cache,naderializadeh2016fundamental,7996347}, among others. Many of these extensions share similar ideas in terms of the achievability and the converse bounds. Therefore, if we can improve the results for the basic bottleneck caching network, the ideas can be used to improve the results in other cases as well.

  In the literature, various approaches have been proposed to improve the bounds on rate-memory tradeoff for the bottleneck network. 
  Several caching schemes have been proposed in \cite{DBLP:journals/corr/Chen14h, wan2016caching, sahraei2016k, tian2016caching, amiri2016fundamental, amiri2016coded, yu2016exact, gomez2016fundamental}, and converse bounds have also been introduced in \cite{ghasemi15,   lim2016information, sengupta15, DBLP:journals/corr/WangLG16, tian2016symmetry, prem2015critical}.
  For the case, where the prefetching is uncoded, the exact rate-memory tradeoff for both peak and average rate (under uniform file popularity) and for both centralized and decentralized settings have been established in~\cite{yu2016exact}. However, for the general case, where the cached content can be an arbitrary function of the files in the database, the exact characterization of the tradeoff remains open. In this case, the state of the art is an approximation within a factor of $4$  for peak rate \cite{ghasemi15} and $4.7$ for average rate under uniform file popularity~\cite{lim2016information}.

    In this paper, we improve the approximation on characterizing the  rate-memory tradeoff by proving new information-theoretic converse bounds, 
     and achieving an approximation 
     within a factor of $2.00884$, for both the \emph{peak rate} and the \emph{average rate} under uniform file popularity. These converse bounds hold for
        the general information theoretic framework, in the sense that there is no constraint on the caching or delivery process. In particular it is not limited to linear coding or uncoded prefetching. 
     This improved characterization is approximately a two-fold improvement with respect to the state of the art in  current literature~\cite{ghasemi15,lim2016information}.

          Furthermore,  for a practically important case where the number of files is large, we exactly characterize the rate-memory tradeoff for systems with no more than $5$ users. In this case, we also  characterize the rate-memory tradeoff within a factor of $2$ for networks with an arbitrary number of users, slightly improving our factor-of-$2.00884$ characterization in the general case.
     In prior works, despite various attempts, this tradeoff has only been exactly characterized in two instances: the single-user case~\cite{maddah-ali12a} and, more recently, the two-user case~\cite{tian2016symmetry}.

    To prove these results we develop two new converse bounds for cache networks. The first converse is developed based on the idea of enhancing the cutset bound, to effectively capture the isolation of cache contents of the users that belong to the same side of the cut. 
  This approach strictly improves the compound cutset bound, which was used in most of the prior works. Furthermore, using this converse, we are able to characterize both the peak rate and the average rate within factor of $2.00884$. To prove this result, we essentially demonstrate that our new converse is within a factor of $2.00884$ from the achievable scheme developed in [21] for all possible parameter values.
    
   Moreover, we develop a second converse bound, which is proved by carefully dividing the set of all user demands into certain subsets, and lower bounding the communication rate within each subset separately. Unlike the first converse, it exploits the scenarios where users may have common demands. This enables improvement upon the first converse, and allows exact characterization of the rate-memory tradeoff for systems with up to $5$ users.

    The rest of this paper is organized as follows.  In Section \ref{sec:sys}, we formally define the caching framework and the rate-memory tradeoff. Then in Section \ref{sec:main} we summarize our main results.  Section \ref{sec:ge} proves our first main result, which characterizes the peak rate-memory tradeoff within a constant factor of $2.00884$ for all possible parameter values, and characterizes this tradeoff within a factor of $2$ when the number of files is large. Section \ref{sec:ge} proves the converse bound that is needed to establish this characterization. For brevity, we prove the rest of the results in appendices.

	\section{System Model and Problem Formulation}\label{sec:sys}
	
In this section, we formally introduce the system model for the caching problem. Then we define the rate-memory tradeoff for both peak rate and average rate based on the introduced framework, and state the corresponding main problems studied in this paper.
\subsection{System Model}
	
	We consider a system with one server connected to $K$ users through a shared, error-free link (see Fig. \ref{fig:net}). The server has access to a database of $N$ files $1, . . . , N$, each of size $F$ bits. We assume that the contents of all files, denoted by $W_1,...,W_N$, are i.i.d. random variables, each of which is uniformly distributed on set $\{1,...,2^F\}$.
	Each user $k$ has an isolated cache memory of size $MF$ bits, where $M\in[0,N]$. For convenience, we define a parameter $r=\frac{KM}{N}$.
	
		\begin{figure}[htbp]
			\centering
			\includegraphics[width=0.5\textwidth]{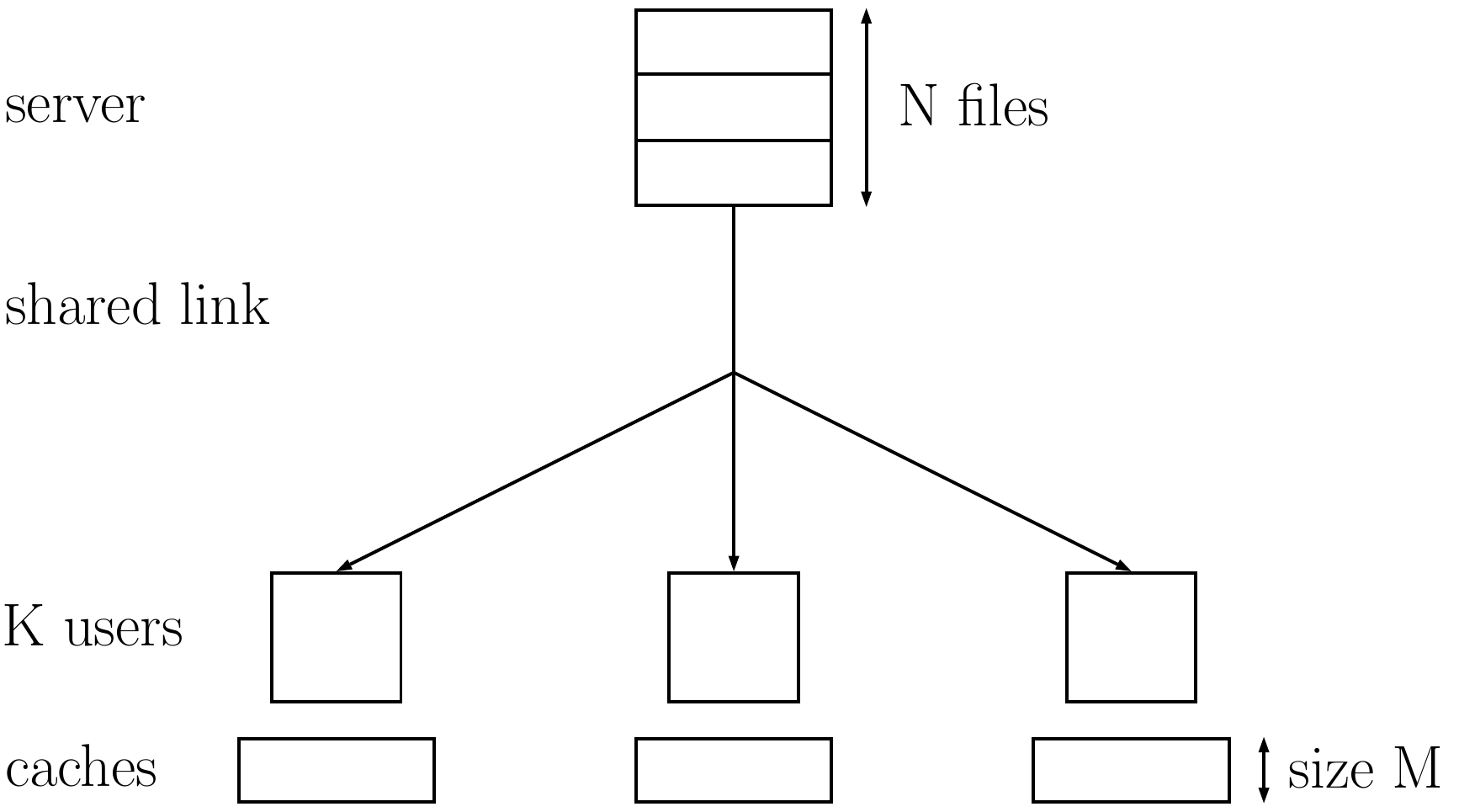}
			\caption{Caching system considered in this paper. The figure illustrates the case where $K=N=3$, and $M=1$.}
			\label{fig:net}
		\end{figure}
		
	The system operates in two phases: a placement phase and a delivery phase. 
	In the placement phase, the users are given access to the entire database. 
 Each user can fill the contents of their caches using the database without knowledge of their future demands.\footnote{This is due to the fact that in most caching systems the caching phase happens during off-peak hours, in order to improve performance during the peak hours when actual user demands are revealed.} We denote the cached content of each user $k$ by $Z_k$.
	Then in a following delivery phase,
	 only the server has access to the
	database of files, and each user requests one of the files 
	in the database. To characterize the requests from the users, we define \textit{demand} $\boldsymbol{d}=\left(d_1,...,d_K\right)$, where $d_k$ is the file requested by user $k$.  
		
	The server is informed of
	the demand and proceeds by generating a message of size $RF$ bits, denoted by $X_{\boldsymbol{d}}
	$, as a function of $W_1,...,W_N$, and sends the message over the shared link. $R$ is a fixed real number given the demand $\boldsymbol{d}$. 
	The quantities $RF$ and $R$ are referred to as the load and the rate of the shared
	link, respectively. Using the contents $Z_k$ of its cache and the message $X_{\boldsymbol{d}}$ received over the shared
	link, each user $k$ aims to reconstruct its requested file $W_{d_k}$.
	
	\subsection{Problem Definition}
	
	Based on the above framework, we define the rate-memory tradeoff using the following terminology.
	We characterize a prefetching scheme by its $K$ caching functions $\boldsymbol{\phi}=(\phi_1,...,\phi_K)$, each of which maps the file contents to the cache content of a specific user:
	\begin{align}
	Z_k=\phi_k(W_1,...,W_N) ~~~~~~~~ \forall k\in\{1,...,K\}.
	\end{align}
	Given a prefetching scheme $\boldsymbol{\phi}$, we say that a communication rate $R$ is \textit{$\epsilon$-achievable} if and only if, for every request $\boldsymbol{d}$, there exists a message $X_{\boldsymbol{d}}$ of length $RF$ that allows all users to recover their desired file $d_k$ with a probability of error of at most $\epsilon$. 
	Given parameters $N$, $K$, and $M$, we define the minimum peak rate, denoted by $R^*$, as the minimum rate that is $\epsilon$-achievable over all prefetching schemes for large $F$ and any $\epsilon>0$. Rigorously,
		\begin{align}
	R^*=\sup_{\epsilon>0} \limsup_{F\rightarrow\infty}&\min_{\boldsymbol{\phi}}\{R\ |\nonumber\\ R&\textup{ is $\epsilon$-achievable given prefetching } \boldsymbol{\phi} \} 
	\end{align}

		Similarly for the average rate, we say that a communication rate $R$ is \textit{$\epsilon$-achievable for demand} $\boldsymbol{d} $, given a prefetching scheme $\boldsymbol{\phi}$, if and only if we can create a message $X_{\boldsymbol{d}}$ of length $RF$ that allows all users to recover their desired file $d_k$ with a probability of error of at most $\epsilon$. 
	Given parameters $N$, $K$, and $M$, we define the minimum average rate, denoted by $R^*_{\textup{ave}}$, as the minimum rate over all prefetching schemes such that, we can find a function $R(\boldsymbol{d})$ that is 
	is $\epsilon$-achievable for any demand $\boldsymbol{d}$, satisfying $R^*_{\textup{ave}}=\mathbb{E}_{\boldsymbol{d}}[R(\boldsymbol{d})]$, where $\boldsymbol{d}$ is uniformly random in $\mathcal{D}=\{1,...,N\}^K$, for large $F$ and any $\epsilon>0$.

Finding the rate-memory tradeoff is essentially finding the values of $R^*$ and  $R^*_{\textup{ave}}$ as a function of $N$, $K$, and $M$. In this paper, we aim to find converse bounds that characterize  $R^*$ and $R^*_{\textup{ave}}$ within a constant factor. 
Moreover, we aim to  better characterize $R^*$ and $R^*_{\textup{ave}}$ for an important case where $N$ ls large, when $K$ and $\frac{M}{N}$ are fixed.

\subsection{Related Works}

Coded caching was originally proposed in \cite{maddah-ali12a}, where the peak rate vs. memory tradeoff was characterized within a factor of $12$. This result was later extended in \cite{niesen13}, where the minimum average rate under uniform file popularity was characterized within a factor of 72. Since then, various efforts has been made on improving these characterizations \cite{ghasemi15, lim2016information, sengupta15, DBLP:journals/corr/WangLG16}. The state of the art is an approximation within a factor of $4$  for peak rate \cite{ghasemi15} and $4.7$ for average rate~\cite{lim2016information}.

In this paper, we characterize both the peak rate and the average rate within a factor of $2.00884$, which is about a two-fold improvement upon the prior arts. This improvement is achieved by improving both the achievability scheme and the converse. Specifically, we use the achievability scheme we recently proposed in \cite{yu2016exact} to upper bound the communication rates. This upper bound strictly improves upon the communication rates achieved by \cite{maddah-ali12a} (and its relaxed version in \cite{maddah-ali13}), which was relied on by all the above works (i.e., \cite{niesen13, ghasemi15, lim2016information, sengupta15, DBLP:journals/corr/WangLG16}). It also achieves the exact optimum communication rates among all caching schemes with uncoded prefetching, for all possible values of $N$, $K$, and $M$. As a shorthand notation, we denote the peak and average rates achieved in \cite{yu2016exact} by $R_{\textup{u}}(N,K,r)$ and $R_{\textup{u,ave}}(N,K,r)$, respectively.\footnote{Recall that $r\triangleq\frac{KM}{N}$. The letter ``u'' in the subscript represents ``upper bound'', and ``uncoded prefetching''.} More precisely, we define these functions as follows.

				\begin{definition}\label{def}
						Given problem parameters $N$, $K$, $M$, and $r=\frac{KM}{N}$, we define
			\begin{align}
			    	R_{\textup{u}}(N,K,r)&=
				\frac{\binom{K}{r+1}-\binom{K-\min\{K,N\}}{r+1}}{\binom{K}{r}},\label{eq:rudef}\\
					R_{\textup{u,ave}}(N,K,r)&=\mathbb{E}_{\boldsymbol{d}}\left[
		\frac{\binom{K}{r+1}-\binom{K-N_{\textup{e}}(\boldsymbol{d})}{r+1}}{\binom{K}{r}}\right]\label{eq:arudef}
			\end{align}
				for $r\in \{0,...,K\}$,  where $\boldsymbol{d}$ is uniformly random in $\mathcal{D}=\{1,...,N\}^K$, and $N_{\textup{e}}(\boldsymbol{d})$ denotes the number of distinct requests in $\boldsymbol{d}$.\footnote{Here the letter ``e'' in the subscript represents ``effective'', given that the function $N_{\textup{e}}(\boldsymbol{d})$ can also be interpreted as the ``effective'' number of files for any demand $\boldsymbol{d}$. Specifically, for any demand $\boldsymbol{d}$, the needed communication rate stated in equation (\ref{eq:arudef}) is exactly the peak communication rate stated in equation (\ref{eq:rudef}) for a caching system with $N=N_{\textup{e}}(\boldsymbol{d})$ files.}   
				Furthermore, for general (non-integer) $r\in[0,K]$, $R_{\textup{u}}(N,K,r)$ and $R_{\textup{u,ave}}(N,K,r)$ are defined as the lower convex envelope of their values at $ r\in\{0,1,...,K\}$, respectively. Specifically, for any non-integer $r\in[0,K]$, we have\footnote{Rigorously, the fact that equations (\ref{eq:genrudef}) and (\ref{eq:genarudef}) define lower convex envelopes is due to the convexity of $R_{\textup{u}}(N,K,r)$ and $R_{\textup{u,ave}}(N,K,r)$ on $r\in\{0,1,...,K\}$. This convexity was observed in \cite{yu2016exact} and can be proved using elementary combinatorics. A short proof of the convexity of $R_{\textup{u}}(N,K,r)$ and $R_{\textup{u,ave}}(N,K,r)$ can be found in Appendix \ref{app:convexity}. }
					\begin{align}
			   	R_{\textup{u}}(N,K,r)&=\nonumber\\
			   	(r-\lfloor r\rfloor)& R_{\textup{u}}(N,K,\lceil r\rceil)+(\lceil r\rceil-r)  R_{\textup{u}}(N,K,\lfloor r\rfloor),\label{eq:genrudef}\\
					R_{\textup{u,ave}}(N,K,r)&=\nonumber\\(r-\lfloor r\rfloor)& R_{\textup{u,ave}}(N,K,\lceil r\rceil)+(\lceil r\rceil-r)  R_{\textup{u,ave}}(N,K,\lfloor r\rfloor).\label{eq:genarudef}
			\end{align}
		\end{definition}
		Given the above upper bounds, we develop improved converse bounds in this paper, which provides better characterizations for both the peak rate and the average rate.

	\section{Main Results}	\label{sec:main}

			We summarize our main results in the following theorems.

		\begin{theorem}\label{th:gconst}
			For a caching system with $K$ users, a database of $N$ files, and a local cache size of $M$ files at each user, we have
			\begin{align}
			\frac{R_{\textup{u}}(N,K,r)}{2.00884}&\leq  R^*\leq R_{\textup{u}}(N,K,r),\label{eq:ru}\\
			\frac{R_{\textup{u,ave}}(N,K,r)}{2.00884}&\leq  R^*_{\textup{ave}}\leq R_{\textup{u,ave}}(N,K,r).\label{eq:ruave}
			\end{align}
			where $R_{\textup{u}}(N,K,r)$ and $R_{\textup{u,ave}}(N,K,r)$ are defined in Definition \ref{def}.
		Furthermore, if $N$ is sufficiently large (specifically, $N\geq \frac{K(K+1)}{2}$), we have 
		\begin{align}
			\frac{R_{\textup{u}}(N,K,r)}{2}&\leq  R^*\leq R_{\textup{u}}(N,K,r),\label{eq:lru}\\
			\frac{R_{\textup{u,ave}}(N,K,r)}{2}&\leq  R^*_{\textup{ave}}\leq R_{\textup{u,ave}}(N,K,r).\label{eq:lruave}
			\end{align}
		\end{theorem}

			\begin{remark}
				The above theorem characterizes $R^*$ and $R^*_{\textup{ave}}$ within a constant factor of $2.00884$ for all possible values of parameters $K$, $N$, and $M$. To the best of our knowledge, this gives the best characterization to date. Prior to this work, the best proved constant factors were $4$ for peak rate \cite{ghasemi15}  and $4.7$ for average rate (under uniform file popularity) \cite{lim2016information}. Furthermore, Theorem \ref{th:gconst} characterizes $R^*$ and $R^*_{\textup{ave}}$ for large $N$ within a constant factor of $2$.
		        \end{remark}
		        
		        	\begin{remark}		
			The converse bound that we develop for proving Theorem \ref{th:gconst} also immediately results in better approximation of rate-memory  tradeoff in other scenarios, such as online caching \cite{pedarsani13}, caching with non-uniform demands \cite{niesen13}, and hierarchical caching \cite{ karamchandani14}. For example, in the case of online caching \cite{pedarsani13}, where the current approximation result is within a multiplicative factor of $24$, it can be easily shown that this factor can be reduced to  $4.01768$ using our proposed bounding techniques.
		\end{remark}
		
				\begin{remark}\label{rem:ruave}\label{rem:ru}
		 $R_\textup{u}(N,K,r)$ and $R_\textup{u,ave}(N,K,r)$, as defined in Definition \ref{def}, are the optimum peak rate and the optimum average rate that can be achieved using uncoded prefetching, as we proved in \cite{yu2016exact}. This indicates that for the coded caching problem, using uncoded prefetching schemes is within a factor of $2.00884$ optimal for both peak rate and average rate. More interestingly, we can show that even for the improved decentralized scheme we proposed in \cite{yu2016exact}, where each user fills their cache independently without coordination but the delivery scheme was designed to fully exploit the commonality of user demands, the optimum rate is still achieved  within a factor of $2.00884$ in general, and a factor of $2$ for large $N$. \footnote{This can be proved based on the fact that, in the proof of Theorem \ref{th:gconst}, we showed the communication rates of the  decentralized caching scheme we proposed in \cite{yu2016exact} (e.g., $R_{\textup{dec}}(M)$ for the peak rate) are within constant factor optimal as intermediate steps.}
		\end{remark}

		        	\begin{remark} \label{rem:twoave}
		Based on the proof idea of Theorem \ref{th:gconst}, we can completely characterize the rate-memory tradeoff for the two-user case, for any possible values of $N$ and $M$, for both peak rate and average rate.
		Prior to this work, the peak rate vs. memory tradeoff for the two-user case was characterized in \cite{maddah-ali12a} for $N\leq 2$, and is characterized in \cite{tian2016symmetry} for $N\geq 3$ very recently. However the average rate vs. memory tradeoff has never been completely characterized for any non-trivial case. In this paper, we prove that the exact optimal tradeoff for the average rate for two-user case can be achieved using the caching scheme we provided in \cite{yu2016exact} (see Appendix \ref{app:ave2}).
		\end{remark}

	    To prove the Theorem \ref{th:gconst}, we derive new converse bounds of  $R^*$ and $R^*_{\textup{ave}}$ for all possible values of $K$, $N$, and $M$. We highlight the converse bound of $R^*$ in the following theorem:
		
				\begin{theorem}
				\label{th:general}
			For a caching system with $K$ users, a database of $N$ files, and a local cache size of $M$ files at each user, $R^*$ is lower bounded by
			\begin{equation}
			R^*\geq s-1+\alpha-\frac{s(s-1)-\ell(\ell-1)+2\alpha s}{2(N-\ell+1)} M, \label{eq:th2}
			\end{equation}
			for any $s\in\{1,...,\min\{N,K\}\}$, $\alpha\in[0,1]$, where $\ell\in\{1,...,s\}$ is the minimum value such that\footnote{Such $\ell$ always exists, because when $\ell=s$, (\ref{eq:sml}) can be written as $\alpha s\leq (N-s+1)s$, which always holds true. }
			\begin{equation}\label{eq:sml}
			{\frac{s(s-1)-\ell(\ell-1)}{2}+\alpha s}\leq  (N-\ell+1)\ell.
			\end{equation}
		\end{theorem}
		
		\begin{remark}
		The above theorem improves the state of the art in various scenarios. For example, when $N$ is sufficiently large (i.e., $N\geq \frac{K(K+1)}{2}$), the above theorem gives tight converse bound for $\frac{KM}{N}\leq 1$, 
		as shown in (\ref{eq:46}).
		The above matching converse can not be proved directly using converse bounds provided in \cite{ghasemi15,   lim2016information, sengupta15, DBLP:journals/corr/WangLG16, tian2016symmetry, prem2015critical} (e.g., for $K=4$, $N=10$, and $M=1$, none of these bounds give $R^*\geq 3$).
		\end{remark}

	    \begin{remark}
		Although Theorem \ref{th:general} gives infinitely many linear converse bounds on $R^*$,
		the region of the memory-rate pair $(M,R^*)$ characterized by Theorem \ref{th:general} has a simple shape with finite corner points. Specifically, by applying the arguments used in the proof of Theorem \ref{th:gconst}, one can show that the exact bounded region given by Theorem \ref{th:general} is bounded by the lower convex envelop of points $\{(\frac{N-\ell+1}{s}, \frac{s-1}{2}+\frac{\ell(\ell-1)}{2s})\ |\ s\in\{1,...,J\}, \ell\in\{1,...,s\} \}\cup \{({0,J})\}$, where $J=\min\{N,K\}$. 
		\end{remark}

			For the case of large $N$, we can exactly characterize the values of $R^*$ and  $R^*_{\textup{ave}}$ for $K\leq 5$. We formally state this result in the following theorem:
		\begin{theorem}\label{th:exact5}
			For a caching system with $K$ users, a database of $N$ files, and a local cache size of $M$ files at each user, we have
			\begin{equation}
		R^*=R^*_{\textup{ave}}=R_{\textup{u}}(N,K,r)
			\end{equation}
			for large $N$ (i.e., $N\rightarrow+\infty$) when $K\leq 5$, where $R_\textup{u}(N,K,r)$ is defined in Definition \ref{def}.\footnote{Rigorously, we show that the maximum possible gap between 	$R^*$, $R^*_{\textup{ave}}$, and $R_{\textup{u}}(N,K,r)$ over $M\in[0,N]$ approaches $0$ as $N$ goes to infinity.}
		\end{theorem}

			\begin{remark}
			As discussed in \cite{maddah-ali13}, the special case of large $N$ is important to handle asynchronous demands. More specifically, \cite{maddah-ali13} showed that asynchronous demands can be handled by splitting each file into many subfiles, and delivering concurrent subfile requests using the optimum caching schemes. In this case, we essentially need to solve the caching problem when the number of files (i.e., the subfiles) is large, but the fraction of files that can be stored at each user is fixed. 	In this paper, we completely characterize this tradeoff for systems with up to $5$ users,  for both peak rate and average rate, while in prior works, this tradeoff has only been exactly characterized in two instances: the single-user case~\cite{maddah-ali12a} and, more recently, the two-user case~\cite{tian2016symmetry}. 
			\end{remark}

		\begin{remark}
				Although Theorem \ref{th:exact5} only consider systems with up to $5$ users, the converse bounds used in its proof also tightly characterize the minimum communication rate in  many cases even for systems with more than $5$ users. For both peak rate and average rate, we can show that more than half of the convex envelope achieved by \cite{yu2016exact} are optimal for large $N$ (e.g., see  Lemma \ref{lemma:2} for peak rate).  
		\end{remark}

	    To prove Theorem \ref{th:exact5}, we state the following Theorem, which provides tighter converse bounds on $R^*$ for certain values of $N$, $K$, and $M$.

		\begin{theorem}\label{th:2}
		For a caching system with $K$ users, a database of $N$ files, and a local cache size of $M$ files at each user, $R^*$ is lower bounded by
		\begin{align}\label{eq:thm4}
  R^*\geq \begin{cases}
               \frac{2K-n+1}{n+1}-\frac{K(K+1)}{n(n+1)}\cdot \frac{M}{N}\ \ \ \ \ \  &\textup{if} \ \beta+\alpha\frac{K-2n-1}{2}\leq 0,\\
              \ \\
               \frac{2K-n+1}{n+1}-\frac{2K(K-n)}{n(n+1)}\cdot \frac{M}{N-\beta} &{\textup{otherwise}},
        \end{cases}
\end{align}
for any $n\in\{\max\{1,K-N+1\},...,K-1\}$, where $\alpha=\lfloor \frac{N-1}{K-n} \rfloor$ and $\beta=N-\alpha (K-n)$. 
		\end{theorem}
		
				\begin{remark}
		The above theorem improves Theorem~\ref{th:general} and the state of the art in many cases. 
		For example, when $r\in\left[\left\lceil K-1-\frac{N-1}{\left\lceil\frac{2N}{K+1}\right\rceil} \right\rceil,K-1\right)$, the converse bound (\ref{eq:thm4}) given by $n=\lfloor r+1\rfloor$
		is tight and we have $R^*=R_{\textup{u}}(N,K,r)$. This result can not be proved in general using the converse bounds provided in \cite{ghasemi15,   lim2016information, sengupta15, DBLP:journals/corr/WangLG16, tian2016symmetry, prem2015critical}
		(e.g., for $K=4$, $N=10$, and $M=4$, none of these bounds give $R^*\geq 1$).
		\end{remark}

			\begin{remark}
			We numerically compare our two converse bounds (i.e., Theorem \ref{th:general} and Theorem \ref{th:2}), benchmarked against the upper bound
			$R_{\textup{u}}(N,K,r)$ we achieved in \cite{yu2016exact} under three different settings (see Fig. \ref{fig:compare}). In all these cases, the two converse bounds together provide a tight characterization: Theorem \ref{th:general} is tight for $r\leq 1$ and $r\geq K-1$, and Theorem \ref{th:2} is tight for $1\leq r\leq K-1$. The same holds true in the proof of Theorem
			\ref{th:exact5}, where the number of users is no more than 5 but the number of files is large.
		\end{remark}

 	\begin{figure}[htbp]
\centering
\begin{subfigure}{.45\textwidth}
  \centering
  \captionsetup{justification=centering}
  \includegraphics[width=0.95\linewidth]{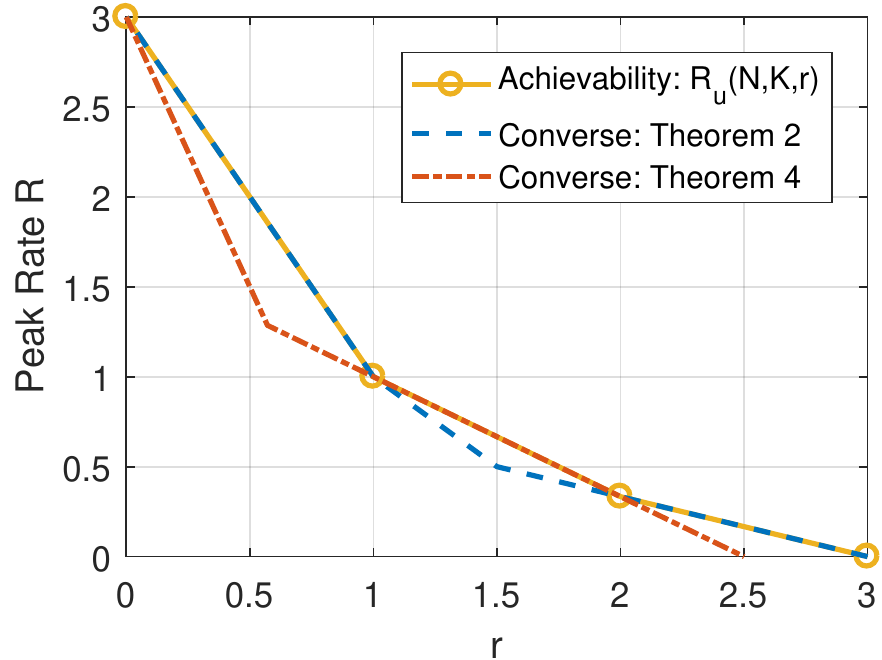}
  \caption{Rate-memory tradeoff for $K=3$, $N=6$.}
\end{subfigure}
\vspace{5mm}

\begin{subfigure}{.45\textwidth}
  \centering
  \captionsetup{justification=centering}
  \includegraphics[width=0.95\linewidth]{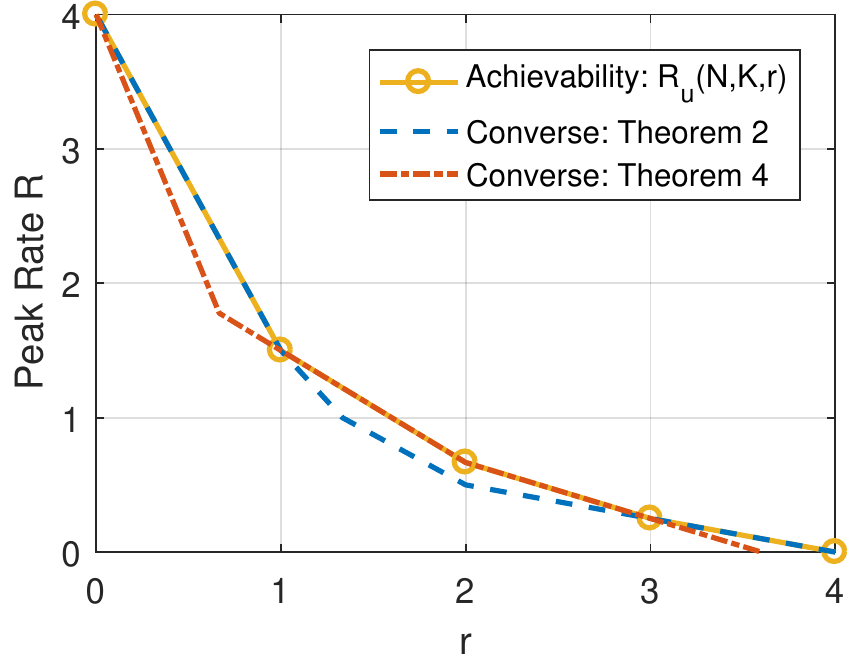}
  \caption{Rate-memory tradeoff for $K=4$, $N=10$. }
\end{subfigure}
\vspace{5mm}

\begin{subfigure}{.45\textwidth}
  \centering
  \captionsetup{justification=centering}
  \includegraphics[width=0.95\linewidth]{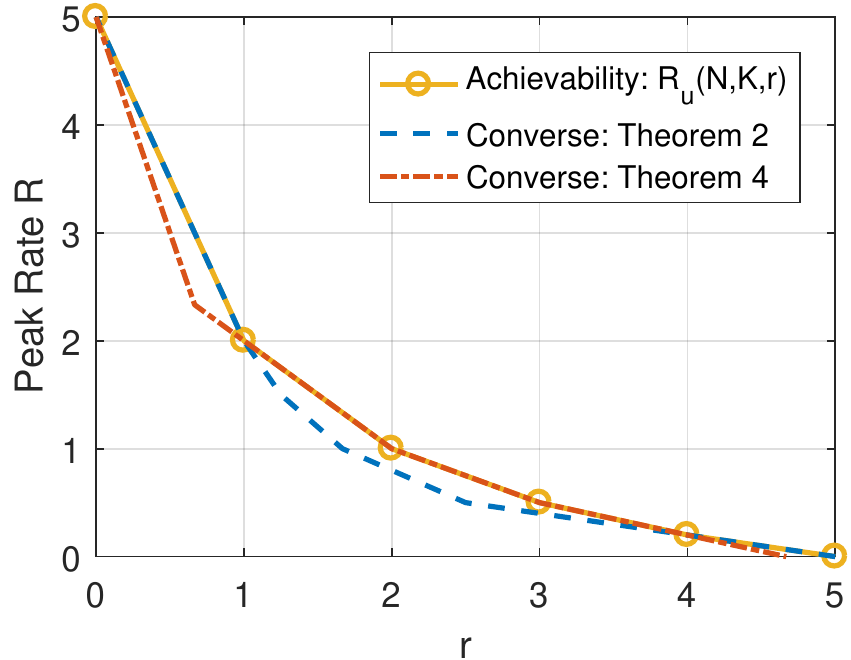}
  \caption{Rate-memory tradeoff for $K=5$, $N\rightarrow +\infty$.  }
\end{subfigure}

\vspace{3mm}
\caption{Numerical comparison among the two converse bounds presented in Theorem \ref{th:general} and Theorem \ref{th:2}, and the upper bound achieved in \cite{yu2016exact}. Our converse bounds tightly characterize the  peak rate-memory tradeoff in all three presented scenarios.}
\label{fig:compare}
\end{figure}

	In the rest of this paper, we prove Theorem \ref{th:gconst} for the peak rate in Section \ref{sec:ge}, and we prove Theorem \ref{th:general} in Section \ref{sec:gconv}.
	For brevity, we prove the rest of the results in the appendices. Specifically,  Appendix \ref{app:proof_tm3p} proves Theorem \ref{th:exact5} for the  peak rate, Appendix \ref{app:proof_tm4} proves Theorem \ref{th:2}, Appendix \ref{app:ave1} proves Theorem \ref{th:gconst} for the average rate, and 
	Appendix \ref{app:ave3} proves Theorem \ref{th:exact5} for the average rate.

	\section{Proof of Theorem \ref{th:gconst} for peak rate} \label{sec:ge}

		In this section, 
	we prove Theorem \ref{th:gconst} assuming the correctness of Theorem \ref{th:general}. The proof of Theorem \ref{th:general} can be found in Section \ref{sec:gconv}. For brevity, we only prove Theorem \ref{th:gconst} for the peak rate (i.e., inequalities (\ref{eq:ru}) and (\ref{eq:lru})) within this section. The proof for the average rate (i.e., inequalities (\ref{eq:ruave}) and (\ref{eq:lruave})) can be found in Appendix \ref{app:ave1}. 	
	
	We start by proving the general factor-of-$2.00884$ characterization for inequality (\ref{eq:ru}). Then we focus on the special case of $N\geq\frac{K(K+1)}{2}$ and prove inequality (\ref{eq:lru}). 
		As mentioned in Remark \ref{rem:ru}, the upper bounds of $R^*$ stated in Theorem \ref{th:gconst} can be proved using the caching scheme provided in \cite{yu2016exact}. Hence, it suffices to prove the lower bounds of (\ref{eq:ru}) and (\ref{eq:lru}).

	\subsection{Proof of inequality (\ref{eq:ru})}
	
	The proof of  inequality (\ref{eq:ru}) consists of $2$ steps. In Step $1$, 
		we first prove, assuming the correctness of Theorem \ref{th:general}, that the memory-rate pair $(M,R^*)$ is lower bounded by the lower convex envelope of a set of points in $\mathcal{S}_{\textup{Lower}}\cup \{({0,J})\}$, where
			\begin{align}\label{def:sl}
		\mathcal{S}_{\textup{Lower}}=&\left\{(M,R)=\left(\frac{N-\ell+1}{s}, \frac{s-1}{2}+\frac{\ell(\ell-1)}{2s}\right)\ \right| \nonumber\\ &s\in\{1,...,J\}, \ell\in\{1,...,s\} \left.\vphantom{(M,R)=\left(\frac{N-\ell+1}{s}, \frac{s-1}{2}+\frac{\ell(\ell-1)}{2s}\right)}\right\}
		\end{align}
		where  $J=\min\{N,K\}$, given parameters  $N$ and $K$.
		Then in Step $2$, we exploit the convexity of the upper bound  $R_{\textup{u}}(N,K,r)$, and prove that it is within a factor of $2.00884$ from the above converse by checking all the corner points of the envelope.
		
		For Step $1$,	we first prove that $R^*$ is lower bounded by the convex envelope.	To prove this statement, it is sufficient to show that any linear function that lower bounds all points in $\mathcal{S}_{\textup{Lower}}\cup \{({0,J})\}$, also lower bounds the point $(M,R^*)$. We prove this for any such linear function, denoted by $A+BM$, by first finding a converse bound of $R^*$ using Theorem \ref{th:general} with certain parameters $s$ and $\alpha$, and then proving that this converse bound is lower bounded by the linear function. 
		We consider the following $2$ possible cases:
		
		If $A\geq 0$, note that $(0,J)$ should be lower bounded by the linear function, so we have $A\leq J$. Thus, we can choose $s=\lceil A\rceil$, $\alpha=A-s+1$, and let $\ell$
		be the minimum value in $\{1,...,s\}$ such that 
		(\ref{eq:sml}) holds. 
		Because $\left(\frac{N-\ell+1}{s},\frac{s-1}{2}+\frac{\ell(\ell-1)}{2s}\right)\in\mathcal{S}_{\textup{Lower}}$, we have 
		\begin{align}
		    A+B\frac{N-\ell+1}{s}\leq \frac{s-1}{2}+\frac{\ell(\ell-1)}{2s}.
		\end{align}
		By the definition of $\alpha$, we have $A=s-1+\alpha$. Consequently, the slope $B$ can be upper bounded as follows:
		\begin{align}
		    B&\leq \frac{s(s-1)+\ell(\ell-1)-2As}{2(N-\ell+1)}\nonumber\\
		    &=-\frac{s(s-1)-\ell(\ell-1)+2\alpha s}{2(N-\ell+1)}.
		\end{align}
		Thus, for any $M\geq0$, we have
		\begin{align}
		    A+BM\leq s-1+\alpha-\frac{s(s-1)-\ell(\ell-1)+2\alpha s}{2(N-\ell+1)}M.
		\end{align}
		Note that the RHS of the above inequality is exactly the lower bound provided in Theorem \ref{th:general}. Hence, $A+BM\leq R^*$.

		If $A<0$, let $s=\ell=1$, we have $(N,0)\in\mathcal{S}_{\textup{Lower}}$ from (\ref{def:sl}). Hence, $A+BN\leq 0$, and for any $M\in[0,N]$ we have 
		\begin{align}
		    A+BM &=\frac{A(N-M)+(A+BN)M}{N} \leq 0. 
		\end{align}
		Obviously $R^*\geq 0$, hence we have $A+BM\leq R^*$.
		
		Combining the above two cases, we have proved that the memory-rate pair $(M,R^*)$ is lower bounded by the lower convex envelope of $\mathcal{S}_{\textup{Lower}}\cup \{({0,J})\}$. This completes the proof of Step $1$.
		
		For Step $2$, we only need to prove that the ratio of $R_{\textup{u}}(N,K,r)$ to the lower convex envelope of $\mathcal{S}_{\textup{Lower}}\cup \{({0,J})\}$ is at most $2.00884$. 
		As mentioned at the beginning of this proof, given that the upper bound $R_{\textup{u}}(N,K,r)$ is convex,\footnote{A short proof can be found in Appendix \ref{app:convexity}} this ratio can only be maximized at the corner points of the envelope, which is a subset of $\mathcal{S}_{\textup{Lower}}\cup \{({0,J})\}$.  Hence, we only need to check that  $R_{\textup{u}}(N,K,r)\leq 2.00884 R$ holds for any $(M,R)\in \mathcal{S}_{\textup{Lower}}\cup \{({0,J})\}$.
		
		To further simplify the problem, we upper bound $R_{\textup{u}}(N,K,r)$ using the following inequality, which can be easily proved using the results of \cite{yu2016exact}:\footnote{Here the upper bound $R_{\textup{dec}}(M)$ is the exact minimum communication rate needed for decentralized caching with uncoded prefetching, as proved in  \cite{yu2016exact}. When $M=0$, $R_{\textup{dec}}(M)\triangleq J$.}
			\begin{equation}\label{eq:decdef}
		    R_\textup{u} (N,K,r)\leq R_{\textup{dec}}(M)\triangleq\frac{N-M}{M} (1-(1-\frac{M}{N})^J). 
		\end{equation}
		Consequently, to prove inequality (\ref{eq:ru}), it suffices to prove the following lemma.
\begin{lemma}\label{lemma:decconst}
For any $(M,R)\in \mathcal{S}_{\textup{Lower}}\cup \{({0,J})\}$, we have ${R_{\textup{dec}}(M)} \leq  {2.00884}R$.
\end{lemma}
The proof of Lemma \ref{lemma:decconst} can be found in Appendix \ref{app:proofdeclemma}. Assuming its correctness, we have $R_{\textup{u}}(N,K,r)\leq 2.00884 R^*$ for all possible parameter values of $N$, $K$, and $M$. This completes the proof of inequality (\ref{eq:ru}).

		\subsection{Proof of inequality (\ref{eq:lru})}				
	   Now we prove that $R^*\geq \frac{R_{\textup{u}}(N,K,r)}{2}$ holds for any $N\geq \frac{K(K+1)}{2}$.
	   In this case, we can verify that inequality (\ref{eq:sml}) holds for any $s\in\{1,...,K\}$, $\alpha=1$, and $\ell=1$. Consequently, from Theorem $\ref{th:general}$, $R^*$ can be bounded as follows:
		 	\begin{align}
		 	R^*&\geq s-1+1-\frac{s(s-1)+2s}{2(N-1+1)}M\nonumber\\
		 	&=s-\frac{s^2+s}{2}\cdot \frac{M}{N}.\label{eq:largeN}
		 	\end{align}
				Then we prove $R^*\geq \frac{R_{\textup{u}}(N,K,r)}{2}$ by considering the following $2$ possible cases:
		    	If $\frac{KM}{N}\leq1$, we have
		    	\begin{align}
		    	    R_{\textup{u}}(N,K,r)=K-\frac{K^2+K}{2}\cdot \frac{M}{N}
		    	\end{align}
		    	 as defined in Definition \ref{def}. Let $s=K$, we have the following bounds from (\ref{eq:largeN}) which tightly characterizes $R_{\textup{u}}(N,K,r)$:
		    		\begin{equation}
		    		R^*\geq K-\frac{K^2+K}{2}\cdot \frac{M}{N}=R_{\textup{u}}(N,K,r)\geq  \frac{R_{\textup{u}}(N,K,r)}{2}.\label{eq:46}
		    		\end{equation}
		    	If $\frac{KM}{N}> 1$,
		    	let $s=\lfloor \frac{N}{M} \rfloor$, we have $\frac{M}{N}\in[\frac{1}{s+1},\frac{1}{s}]$. Consequently,    	
		    	we can derive the following lower bound on $R^*$:
		    		\begin{align}
		    	R^*&\geq s-\frac{s^2+s}{2}\cdot \frac{M}{N}\nonumber\\
		    	&=\frac{N-M}{2M}+ \frac{s^2+s}{2}\cdot\frac{N}{M}\cdot \left(\frac{M}{N}-\frac{1}{s+1}\right)\cdot\left(\frac{1}{s}-\frac{M}{N}\right)\nonumber\\
		    	&\geq \frac{N-M}{2M}.
		    	\end{align}
		    	As mentioned earlier in this section, the following upper bound can be easily proved using the results of \cite{yu2016exact}:
		    	\begin{align}
		    	R_{\textup{u}}(N,K,r)&\leq \frac{N-M}{M} (1-(1-\frac{M}{N})^K).
		    	\end{align}
		    Consequently, we have $R_{\textup{u}}(N,K,r)\leq \frac{N-M}{M}\leq 2R^* $.
		    	To conclude, we have proved $R^*\geq \frac{R_{\textup{u}}(N,K,r)}{2}$ for both cases. Hence, inequality (\ref{eq:lru}) holds for large $N$ for any possible values of $K$ and $M$.

	\section{Proof of Theorem \ref{th:general}}\label{sec:gconv}
	    Before proving the converse bound stated in Theorem \ref{th:general}, we first present the following key lemma, which gives a lower bound on any $\epsilon$-achievable rate given any prefetching scheme.
			\begin{lemma}\label{lemma:ga}
				Consider a coded caching problem with parameters $N$ and $K$. Given a certain prefetching scheme, for any demand $\boldsymbol{d}$, any $\epsilon$-achievable rate $R$ is lower bounded by \footnote{By an abuse of notation, we denote a sub-array by using a set of indices as the subscript. Besides, we define $\{d_1,...,d_{k-1}\}= \emptyset$ for $k=1$. Similar convention will be used throughout this paper.} 
				\begin{align}
				R\geq &\frac{1}{F}\left(\sum_{k=1}^{\min\{N,K\}} H(W_{d_k}|Z_{\{1,...,k\}}, W_{\{d_1,...,d_{k-1}\}})\right)\nonumber\\&- \min\{N,K\}(\frac{1}{F}+\epsilon).
				\end{align}	
			\end{lemma} 
	The above lemma is developed based on the idea of enhancing the cutset bound, which is further explained in the proof of this lemma in Appendix \ref{app:lemmaga}. One can show that this approach strictly improves the compound cutset bound, which was used in most of the prior works.  We now continue to prove Theorem \ref{th:general} assuming the correctness of Lemma \ref{lemma:ga}.
	
	The rest of the proof consists of two steps. In Step $1$, we exploit the homogeneity of the problem, and derive a symmetrized version of the converse presented in Lemma \ref{lemma:ga}. Then in Step $2$, we derive the converse bound in Theorem \ref{th:general}, which is independent of the prefetching scheme, by essentially minimize the symmetrized converse over all possible designs.
	
		For Step $1$, we observe that the caching problem proposed in this paper assumes that all users has the same cache size, and all files are of the same size. To fully utilize this homogeneity, we define the following useful notations. 		
			For any positive integer $i$, we denote the set of all permutations of $\{1,...,i\}$ by $\mathcal{P}_i$. For any set $\mathcal{S}\subseteq \{1,...,i\}$ and any permutation $p\in \mathcal{P}_i$, we define 
			   $ p\mathcal{S}=\{p(s)\ |\ s\in{S}\}$. For any subsets $\mathcal{A}\subseteq \{1,...,N\}$ and $\mathcal{B}\subseteq \{1,...,K\}$, we define
			\begin{align}
			    H^*(W_{\mathcal{A}},Z_{\mathcal{B}})\triangleq \frac{1}{N!K!} \sum_{p\in \mathcal{P}_N, q\in \mathcal{P}_K} H(W_{p\mathcal{A}},Z_{q\mathcal{B}}).
			\end{align}
			Similarly, we define the same notation for conditional entropy in the same way. We can verify that the functions defined above satisfies all Shannon's inequalities. I.e., for any sets of random variables $\mathcal{A}$, $\mathcal{B}$ and $\mathcal{C}$, we have
			\begin{align}
			    H^*(\mathcal{A}|\mathcal{B})\geq H^*(\mathcal{A}|\mathcal{B}, \mathcal{C}).
			\end{align}
			
			Note that from the homogeneity of the problem, for any $\epsilon$-achievable rate $R$, Lemma \ref{lemma:ga} holds for any demands, under any possible relabeling of the users. Thus, by considering the class of demands where at least $\min\{N,K\}$ files are requested, we have
			\begin{align}
				R\geq &\frac{1}{F}(\sum_{k=1}^{\min\{N,K\}} H(W_{q(k)}|Z_{\{p(1),...,p(k)\}}, W_{\{q(1),...,q(k-1)\}}))\nonumber\\&- \min\{N,K\}(\frac{1}{F}+\epsilon)
				\end{align}	
			 for any $p\in \mathcal{P}_K$ and $q\in \mathcal{P}_N$. Averaging the above bound over all possible $p$ and $q$, we have
				\begin{align}
				R\geq &\frac{1}{F}(\sum_{k=1}^{\min\{N,K\}} H^*(W_{k}|Z_{\{1,...,k\}}, W_{\{1,...,k-1\}}))\nonumber\\&- \min\{N,K\}(\frac{1}{F}+\epsilon).
				\end{align}
		Recall that $R^*$ is defined to be the minimum $\epsilon$-achievable rate over all prefetching scheme $\phi$ for large $F$ for any $\epsilon>0$, we have
				\begin{align}
				R^*\geq&\sup_{\epsilon>0} \limsup_{F\rightarrow\infty}\min_{\boldsymbol{\phi}}\nonumber\\&\{ \frac{1}{F}(\sum_{k=1}^{\min\{N,K\}} H^*(W_{k}|Z_{\{1,...,k\}}, W_{\{1,...,{k-1}\}}))\nonumber\\&- \min\{N,K\}(\frac{1}{F}+\epsilon)\}\nonumber\\
				=&\sup_{\epsilon>0} \limsup_{F\rightarrow\infty}\min_{\boldsymbol{\phi}}\nonumber\\&\{ \frac{1}{F}(\sum_{k=1}^{\min\{N,K\}} H^*(W_{k}|Z_{\{1,...,k\}}, W_{\{1,...,{k-1}\}})\}\nonumber\\
				\geq& \inf_{F\in\mathbb{N_+}}\min_{\boldsymbol{\phi}}\nonumber\\&\{ \frac{1}{F}(\sum_{k=1}^{\min\{N,K\}} H^*(W_{k}|Z_{\{1,...,k\}}, W_{\{1,...,{k-1}\}})\}. 
				\end{align}
				Now we have derived a symmetrized version of the converse bound. To simplify the discussion, we define $R_{\textup{A}}(F,\boldsymbol{\phi})=\frac{1}{F}\sum\limits_{k=1}^{\min\{N,K\}} H^*(W_{k}|Z_{\{1,...,k\}}, W_{\{1,...,{k-1}\}})$. Consequently,
				\begin{align}
				R^*&\geq \inf_{F\in\mathbb{N_+}}\min_{\boldsymbol{\phi}} R_{\textup{A}}(F,\boldsymbol{\phi}). 
				\end{align}
				
				For Step $2$, as mentioned previously in this proof, to derive the converse bound presented in  Theorem \ref{th:general}, we aim to minimize the symmetrized converse $R_{\textup{A}}(F,\boldsymbol{\phi})$ over all prefetching scheme $\boldsymbol{\phi}$. Moreover, we need to prove that it is no less than the RHS of (\ref{eq:th2})  for any parameters $s$ and $\alpha$. We present the following lemma, which essentially solves this problem.

\begin{lemma}\label{lemma:stg}
For any parameters  $s\in\{1,...,\min\{N,K\}\}$, $\alpha\in[0,1]$, and any prefetching scheme $\boldsymbol{\phi}$, we have
\begin{align}\label{ineq:p2l}
     R_{\textup{A}}(F,\boldsymbol{\phi})\geq s-1+\alpha-\frac{s(s-1)-\ell(\ell-1)+2\alpha s}{2(N-\ell+1)} M,
\end{align}
 where $\ell\in\{1,...,s\}$ is the minimum value such that 
			\begin{equation}
			{\frac{s(s-1)-\ell(\ell-1)}{2}+\alpha s}\leq  (N-\ell+1)\ell.
			\end{equation}
\end{lemma}

The proof of Lemma \ref{lemma:stg} can be found in Appendix \ref{app:plstg}. Note that the lower bound in the above lemma is identical to the converse in Theorem \ref{th:general}. Assuming its correctness, then given any $s$ and $\alpha$, we can bound $R^*$ as follows:
		 	\begin{align}
				R^*&\geq \inf_{F\in\mathbb{N_+}}\min_{\boldsymbol{\phi}}R^*_{\textup{A}}(F,\phi) \nonumber\\
				&\geq (s-1+\alpha)- \frac{s(s-1)-\ell(\ell-1)+2\alpha s}{2(N-\ell+1)}M.
				\end{align}
			This completes the proof of Theorem \ref{th:general}.

\section{conclusion}

In this paper, we developed novel converse bounding techniques for caching networks, and characterized the rate-memory tradeoff of the basic bottleneck caching network within a factor of $2.00884$ for both the peak rate and the average rate. This is approximately a two-fold improvement  with  respect  to  the  state  of  the  art.  We also provided  tight  characterization of rate-memory tradeoff for systems with no more than $5$ users,  when the number of files is large.
		The results of this paper can also be used to improve the approximation of rate-memory tradeoff in several other settings, such as online caching, caching with non-uniform demands, and hierarchical caching.		
		
		\appendices

	\section{Proof of Theorem \ref{th:exact5}  for peak rate}\label{app:proof_tm3p}
					
					In this section, we prove Theorem \ref{th:exact5} assuming the correctness of Theorem \ref{th:2}. 
		The proof of Theorem \ref{th:2} can be found in Appendix \ref{app:proof_tm4}. For brevity, we only prove Theorem \ref{th:exact5} for the peak rate (i.e., $R^*=R_{\textup{u}}(N,K,r)$ for large $N$) within this section. The proof for the average rate (i.e., $R^*_{\textup{ave}}=R_{\textup{u}}{(N,K,r)}$ for large $N$) can be found in Appendix \ref{app:ave1}. 	
	
	As mentioned  previously, the rate $R_{\textup{u}}(N,K,r)$ can be exactly achieved using the caching scheme proposed in \cite{yu2016exact}. Hence, to prove Theorem \ref{th:exact5}, it is sufficient to show that $R^*\geq R_{\textup{u}}(N,K,r)$ for large $N$ (i.e., $N\rightarrow+\infty$) when $K\leq 5$. This statement can be easily proved using the following lemma:
	\begin{lemma}\label{lemma:2}
	For a caching problem with parameters $K$, $N$, and $M$, we have $R^*\geq R_{\textup{u}}(N,K,r)$ for large $N$, if $r\leq 1$ or $r\geq \lceil\frac{K-3}{2}\rceil$.
	\end{lemma}
		Assuming the correctness of Lemma \ref{lemma:2}, and noting that the condition in Lemma \ref{lemma:2} (i.e.,  $r\leq 1$ or $r\geq \lceil\frac{K-3}{2}\rceil$) always holds true for $K\leq 5$, we have $R^*\geq R_{\textup{u}}(N,K,r)$ for large $N$ and for all possible values of $M$, in any caching system with no more than $5$ users.
	Hence, to prove Theorem \ref{th:exact5}, it suffices to prove Lemma \ref{lemma:2}. We prove this lemma as follows, using Theorem \ref{th:general} and Theorem \ref{th:2}.

	\begin{proof}[Proof of Lemma \ref{lemma:2}]

  We start by focusing on two easier cases, $r\leq 1$ and $r\geq K-1$. 
  When $r\leq 1$, the inequality $R^*\geq R_{\textup{u}}(N,K,r)$ is already proved in Section \ref{sec:ge} and given by (\ref{eq:46}), 
  for $N\geq \frac{K(K+1)}{2}$.  When $r\geq K-1$, 
  we have $R^*\geq 1-\frac{M}{N}=R_{\textup{u}}(N,K,r)$, which can be proved by choosing $s=1$ and $\alpha=1$ for Theorem \ref{th:general}.
  Hence, we only need to focus on the case where $r\in \left[\max\{\lceil\frac{K-3}{2}\rceil,1\}, K-1\right)$, and show that for large $N$, the maximum possible gap between $R^*$ and $R_{\textup{u}}(N,K,r)$ approaches $0$.
  
  We prove this result using Theorem \ref{th:2}. Essentially, we need to find parameter $n\in\{1,...,K-1\}$ for Theorem \ref{th:2}, such that the corresponding converse bound approaches $R_{\textup{u}}(N,K,r)$ for large $N$.
  
  Let $n=\lfloor r+1\rfloor$, we have 
  \begin{align}
      R_{\textup{u}}(N,K,r)=\frac{2K-n+1}{n+1}-\frac{K(K+1)}{n(n+1)}\cdot \frac{M}{N}\label{eq:a36}
  \end{align}
  by definition, for sufficiently large $N$ (more specifically, $N\geq K-n+1$). Under the same condition for large $N$, we have  $n\in\{\max\{1,K-N+1\},...,K-1\}$ given $r\in[1,K-1)$. Hence, 
  we can use $n$ as the parameter of Theorem \ref{th:2}. Now we prove the tightness of this converse bound by considering the following two possible cases:
  
  If $n>\frac{K-1}{2}$, we have $K-2n-1<0$. Recall that 
				$\alpha=\lfloor \frac{N-1}{K-n} \rfloor$ and $\beta=N-\alpha (K-n)$. We can prove that 
  when $N$ is sufficiently large 
  (i.e. $N\geq\frac{2(K-n)^2}{2n+1-K}+1$), the condition  $\beta+\alpha\frac{K-2n-1}{2}\leq 0$ is always satisfied. Consequently, 
		\begin{align}
		    R^*\geq \frac{2K-n+1}{n+1}-\frac{K(K+1)}{n(n+1)}\cdot \frac{M}{N}=R_{\textup{u}}(N,K,r).
		\end{align}
    If $n\leq \frac{K-1}{2} $, because we are considering the case where $r\geq \lceil\frac{K-3}{2}\rceil$, we have $n=\frac{K-1}{2}$. Hence, we can verify that $\beta+\alpha\frac{K-2n-1}{2}\leq 0$ does not hold for any $N$. Consequently,
    \begin{align}
		    R^*&\geq \frac{2K-n+1}{n+1}-\frac{2K(K-n)}{n(n+1)}\cdot \frac{M}{N-\beta}\nonumber\\
		    &=\frac{2K-n+1}{n+1}-\frac{K(K+1)}{n(n+1)}\cdot \frac{M}{N-\beta}.\label{eq:a38}
		\end{align}
		As $N$ approaches infinity, $\beta$ is upper bounded by a constant. Hence, we have $ 
		\lim\limits_{N\rightarrow+\infty}\frac{N}{N-\beta}= 1$. Therefore, from (\ref{eq:a36}) and (\ref{eq:a38}), we have
    \begin{align}
		   \lim_{N\rightarrow+\infty} (R^*-R_{\textup{u}}&(N,K,r))\geq\nonumber\\& \lim_{N\rightarrow+\infty}  
		   \frac{K(K+1)}{n(n+1)}\cdot \left(\frac{M}{N}-\frac{M}{N-\beta}\right)\nonumber\\
		   =& \lim_{N\rightarrow+\infty}  \frac{r(K+1)}{n(n+1)}\cdot \left(1-\frac{N}{N-\beta}\right)\nonumber\\
		    =&0.
		\end{align}
			\end{proof}

	\section{Proof of Theorem \ref{th:2}}\label{app:proof_tm4}
         Before proving the converse bounds stated in Theorem \ref{th:2}, we first present the following key lemma, which gives a lower bound on any $\epsilon$-achievable rate given any prefetching scheme.
         	\begin{lemma}\label{lemma:3}
						Consider a coded caching problem with parameters $N$ and $K$. Given a certain prefetching scheme, any $\epsilon$-achievable rate $R$ is lower bounded by \footnote{Here we adopt the notation of $H^*(W_\mathcal{A}, Z_{\mathcal{B}})$ which is defined in the proof of Theorem \ref{th:general}.}
					\begin{align}\label{eq:lemma3}
							    RF\geq& H^*(W_1| Z_1)\nonumber\\&+\frac{2}{n(n+1)\alpha} \left(\vphantom{\sum_{i=0}^{K-n-1} H^*(Z_{1}|  W_{\{1,...,\beta+i\alpha\}})}\alpha n(K-n)F
					  - nH^*(Z_1|W_{\{1,...,\beta\}})\right.\nonumber\\&\left.-\sum_{i=0}^{K-n-1} H^*(Z_{1}|  W_{\{1,...,\beta+i\alpha\}})\right)\nonumber\\
					  &~ -\frac{2K-n+1}{n+1}(1+\epsilon F)
					\end{align}
					 for any integer $n\in\{\max\{1,K-N+1\},...,K-1\}$, 
					where $\alpha=\lfloor \frac{N-1}{K-n} \rfloor$ and $\beta=N-\alpha (K-n)$.
					\end{lemma}
				We postpone the proof of the above lemma to Appendix \ref{app:lm3}, and continue to prove Theorem $\ref{th:2}$ assuming its correctness.	
				To simplify the discussion, we define 
										\begin{align} \label{eq:rbdef}
				R_{\textup{B}}(F,\boldsymbol{\phi})=&\frac{1}{F}\left( H^*(W_1| Z_1)+\frac{2}{n(n+1)\alpha} \left(\vphantom{\sum_{i=0}^{K-n-1}}\alpha n(K-n)F\right.\right.\nonumber\\&
					  - nH^*(Z_1|W_{\{1,...,\beta\}})\nonumber\\&\left.\left.-\sum_{i=0}^{K-n-1} H^*(Z_{1}|  W_{\{1,...,\beta+i\alpha\}})\right)\right).
					\end{align}
				Using Lemma \ref{lemma:3}, we have
					\begin{align}
							    R&\geq R_{\textup{B}}(F,\boldsymbol{\phi})
					   -\frac{2K-n+1}{n+1}(\frac{1}{F}+\epsilon )
					\end{align}
					if $R$ is $\epsilon$-achievable. Recall that $R^*$ is defined to be the minimum $\epsilon$-achievable rate over all prefetching scheme $\phi$ for large $F$ for any $\epsilon>0$, we have the following lower bound on $R^*$:
				\begin{align}\label{eq:63}
				R^*&\geq\sup_{\epsilon>0} \limsup_{F\rightarrow\infty}\min_{\boldsymbol{\phi}}\{ R_{\textup{B}}(F,\boldsymbol{\phi})
					   -\frac{2K-n+1}{n+1}(\frac{1}{F}+\epsilon )\}\nonumber\\
				&=\sup_{\epsilon>0} \limsup_{F\rightarrow\infty}\min_{\boldsymbol{\phi}} R_{\textup{B}}(F,\boldsymbol{\phi})
					   \nonumber\\
				&\geq \inf_{F\in\mathbb{N_+}}\min_{\boldsymbol{\phi}} R_{\textup{B}}(F,\boldsymbol{\phi})	 . 
				\end{align}
					Hence, to prove Theorem $\ref{th:2}$, we only need to prove that for any prefetching scheme $\boldsymbol{\phi}$,  $R_{\textup{B}}(F,\boldsymbol{\phi})$ is lower bounded by the converse bounds given in Theorem $\ref{th:2}$ for any valid parameter $n$.
					
					Now consider any  $n\in\{\max\{1,K-N+1\},...,K-1\}$.
					For brevity, 
					we define 
					\begin{align}\label{eq:thetadef}
					\theta=\left(K\beta+\frac{(K-n)(K-n-1)}{2}\alpha \right).
					\end{align}
				Equivalently, we have 
					\begin{align}
					    \theta=n\beta+\sum\limits_{i=0}^{K-n-1} (\beta+i\alpha).
					\end{align} 
					Hence,
					
					\begin{samepage}
				\begin{align}\label{eq:65}
				    \theta H^*(W_1| Z_1)\geq& nH^*(W_{\{1,...,\beta\}}| Z_1) \nonumber\\&+\sum_{i=0}^{K-n-1} H^*(W_{\{1,...,\beta+i\alpha\}}| Z_1)\nonumber\\
				    =&\theta F+
					   nH^*(Z_1|W_{\{1,...,\beta\}}) \nonumber\\&+\sum_{i=0}^{K-n-1} H^*(Z_{1}|  W_{\{1,...,\beta+i\alpha\}})-K H^*(Z_1).
				\end{align}
				\end{samepage}
				From (\ref{eq:rbdef}) and (\ref{eq:65}), we have
					\begin{align} 
				R_{\textup{B}}(F,\boldsymbol{\phi})F\geq& \left(1- \frac{2\theta}{n(n+1)\alpha}\right)H^*(W_1| Z_1)\nonumber\\&+\frac{2}{n(n+1)\alpha} (\theta F
					   -K H^*(Z_1)\nonumber\\&+\alpha n(K-n)F
					  ).
					\end{align}
				Depending on the value of $\theta$, we bound $H^*(W_1| Z_1)$ in $2$ different ways:
				
				When $	1\geq \frac{2\theta}{n(n+1)\alpha}$, this is exactly the case where  $\beta+\alpha\frac{K-2n-1}{2}\leq 0$ holds.		We use the following bound:
					\begin{align}
				     H^*(W_1| Z_1)\geq F-\frac{H^*(Z_1)}{N}.
				\end{align}
				Consequently, 
					\begin{align} R_{\textup{B}}(F,\boldsymbol{\phi})F\geq& \left(1- \frac{2\theta}{n(n+1)\alpha}\right)\left(F-\frac{H^*(Z_1)}{N}\right)\nonumber\\&+\frac{2}{n(n+1)\alpha} (\theta F
					   -K H^*(Z_1)\nonumber\\&+\alpha n(K-n)F).
					   \end{align}
					   Given $\theta$ defined in (\ref{eq:thetadef}), and $\beta=N-\alpha(K-n)$ as defined in Lemma \ref{lemma:3}, we have
					   \begin{align}
					R_{\textup{B}}(F,\boldsymbol{\phi})F&= \frac{2K-n+1}{n+1}F-\frac{K(K+1)}{n(n+1)}\cdot \frac{H^*(Z_1)}{N}\nonumber\\
					&\geq  \frac{2K-n+1}{n+1}F-\frac{K(K+1)}{n(n+1)}\cdot \frac{M}{N}F.
					\end{align}
					Hence we have the follows from (\ref{eq:63}): 
		 	\begin{align}
				R^*
				&\geq \frac{2K-n+1}{n+1}-\frac{K(K+1)}{n(n+1)}\cdot \frac{M}{N}.
				\end{align}
	
					On the other hand, when $1< \frac{2\theta}{n(n+1)\alpha}$, this is exactly the case where  $\beta+\alpha\frac{K-2n-1}{2}\leq 0$ does not hold.
					We use $H^*(W_1| Z_1)\leq F$.
					Similarly,  
					\begin{align} 
					R_{\textup{B}}(F,\boldsymbol{\phi})F\geq& \left(1- \frac{2\theta}{n(n+1)\alpha}\right)F\nonumber\\&+\frac{2}{n(n+1)\alpha} (\theta F
					   -K H^*(Z_1)\nonumber\\&+\alpha n(K-n)F)\nonumber\\
					  =& \frac{2K-n+1}{n+1}F-\frac{2K(K-n)}{n(n+1)}\cdot \frac{H^*(Z_1)}{N-\beta}\nonumber\\
					\geq&  \frac{2K-n+1}{n+1}F-\frac{2K(K-n)}{n(n+1)}\cdot \frac{M}{N-\beta}F.
					\end{align}
							Hence, 
		 	\begin{align}
				R^*&\geq \inf_{F\in\mathbb{N_+}}\min_{\boldsymbol{\phi}}R_{\textup{B}}(F,\phi) \nonumber\\
				&\geq \frac{2K-n+1}{n+1}-\frac{2K(K-n)}{n(n+1)}\cdot \frac{M}{N-\beta}.
				\end{align}
				To conclude, we have proved that the converse bound given in Theorem \ref{th:2} holds for any valid parameter $n$.
	
\section{Proof of Lemma \ref{lemma:decconst}}\label{app:proofdeclemma}

In this appendix, we aim to prove that for any $(M,R)\in \mathcal{S}_{\textup{Lower}}\cup \{({0,J})\}$,  ${R_{\textup{dec}}(M)} \leq  {2.00884}R$. Note that if $(M,R)=(0,J)$, we have $R_{\textup{dec}}(M)=J\leq2.00884R$. Hence, it suffices to consider the case where $(M,R)\in \mathcal{S}_{\textup{Lower}}$.

In this case, we can find $s\in\{1,...,J\}$ and $\ell\in\{1,...,s\}$ such that 
					\begin{align}\label{eq:corner}
					 (M,R)=\left(\frac{N-\ell+1}{s}, \frac{s-1}{2}+\frac{\ell(\ell-1)}{2s}\right).
					\end{align}
Based on the parameter values, we prove $R_{\textup{dec}}(M)\leq 2.00884 R$ by considering the following $3$ possible scenarios:
		
		\noindent a). If $N\geq 9s$, we first have the follows given (\ref{eq:decdef}):
		\begin{align}
		    R_{\textup{dec}}(M)&\leq \frac{N-M}{M}.
		    \end{align}
		    Due to (\ref{eq:corner}), the above inequality is equivalent to
		    \begin{align}   
		    R_{\textup{dec}}(M)&\leq s-1+\frac{s(\ell-1)}{N-\ell+1}.
		    \end{align}
		    Recall that $s\geq \ell$ and $N\geq 9s$, we have
		    \begin{align}
		    R_{\textup{dec}}(M)&\leq s-1+\frac{s(\ell-1)}{N-s}\nonumber\\
		    &\leq s-1+\frac{\ell-1}{8}.
		\end{align}
		Since $s\geq \ell$, we have $\frac{\ell-1}{\ell}\leq \frac{s-1}{s}$. Consequently, 
		\begin{align}
		    R_{\textup{dec}}(M)&\leq s-1+\frac{\sqrt{\ell-1}}{8} \cdot\sqrt{\frac{(s-1)}{s}\cdot \ell}\nonumber\\
		    &= s-1+2\cdot \sqrt{\frac{{s-1}}{256}}\cdot \sqrt{\frac{\ell(\ell-1)}{s}}.
		    \end{align}
		    Applying the AM-GM inequality to the second term of the RHS, we have
		    \begin{align}
		     R_{\textup{dec}}(M)&\leq  s-1+\frac{s-1}{256}+\frac{\ell(\ell-1)}{s}. 
		     \end{align}
		     Because $\ell\geq 1$, we can thus upper bound $R_{\textup{dec}}(M)$ as a function of $R$, which is given in (\ref{eq:corner}):
		     \begin{align}
		    R_{\textup{dec}}(M)&\leq (2+\frac{1}{128})(\frac{s-1}{2}+\frac{\ell(\ell-1)}{2s})\nonumber\\
		    &\leq 2.00884R.
		\end{align}
		
	\noindent	b). If $N<9s$ and $N\leq 81$, we upper bound $R_{\textup{dec}}(M)$ as follows:
		\begin{align}
		    R_{\textup{dec}}(M)&\leq \frac{N-M}{M}(1-(1-\frac{M}{N})^N).
		\end{align}
		Note that both the above bound and $R$ are functions of $N$, $s$ and $\ell$, which can only take values from $\{1,...,81\}$. Through a brute-force search, we can show that $R_{\textup{dec}}(M)\leq 2.000R\leq 2.00884R$.
		
	\noindent c).	If $N<9s$ and $N>81$, recall that $M=\frac{N-\ell+1}{s}$ from (\ref{eq:corner}), we have
		\begin{align}\label{eq:mbound}
		    M
		    &\leq \frac{N}{s} 
		    < 9.
		\end{align}
		Similarly, $R$ can be lower bounded as follows given (\ref{eq:corner}):
		\begin{align}
		    R&=\frac{s-1}{2}+\frac{(N-sM)(N-sM+1)}{2s}\nonumber\\
		    &=\frac{(1+M^2)s}{2}+\frac{N(N+1)}{2s}-(N+\frac{1}{2})M-\frac{1}{2}.
		    \end{align}
		    Applying the AM-GM inequality to the first two terms of the RHS, we have
		    \begin{align}
		    R &\geq \sqrt{(1+M^2)N(N+1)}-(N+\frac{1}{2})M-\frac{1}{2}.
		\end{align}
		From (\ref{eq:mbound}), $N>81>M^2$, we have $\sqrt{N(N+1)}\geq \sqrt{M^2(M^2+1)}+N-M^2$. Consequently,
		\begin{align}
		    R\geq& \sqrt{1+M^2}(\sqrt{M^2(M^2+1)}+N-M^2)\nonumber\\&-(N+\frac{1}{2})M-\frac{1}{2}\nonumber\\
		    =&(N-81)(\sqrt{1+M^2}-M)\nonumber\\&+(81-M^2)(\sqrt{1+M^2}-M)+\frac{M-1}{2}.
		\end{align}
		On the other hand, we upper bound $R_{\textup{dec}}(M)$ as follows:
			\begin{align}
		    R_{\textup{dec}}(M)&\leq \frac{N-M}{M}(1-(1-\frac{M}{N})^ N)\nonumber\\
		    &= \frac{N-M}{M}(1-e^{\ln (1-\frac{M}{N})N}).
		\end{align}
		From (\ref{eq:mbound}), $\frac{M}{N}<\frac{9}{81}=\frac{1}{9}$, it is easy to show that $\ln (1-\frac{M}{N})\geq -\frac{M}{N}-\frac{9}{16}\left(\frac{M}{N}\right)^2$. Hence,					
			\begin{align}
		    R_{\textup{dec}}(M)\leq& \frac{N-M}{M}(1-e^{-M-\frac{9}{16}\frac{M^2}{N}})\nonumber\\
		    \leq& \frac{N-M}{M}\left(1-e^{-M} (1-\frac{9}{16}\frac{M^2}{N})\right)\nonumber\\
		    \leq& \frac{N-M}{M}(1-e^{-M})+\frac{N}{M} e^{-M}\frac{9}{16}\frac{M^2}{N}\nonumber\\
		    =&(N-81) \frac{1-e^{-M}}{M} +\frac{81-M}{M}({1-e^{-M}})\nonumber\\&+\frac{9}{16}M e^{-M}.
		\end{align}				
		Numerically, we can verify that the following inequalities hold for $M\in [0,9)$:
		\begin{align}
		     \frac{1-e^{-M}}{M}&\leq 2.00884 (\sqrt{1+M^2}-M) ,\label{num:1}\\
	            \frac{81-M}{M}(&{1-e^{-M}})+\frac{9}{16}M e^{-M}  \leq\nonumber\\
	            2.00884&\left((81-M^2)(\sqrt{1+M^2}-M)+\frac{M-1}{2}\right).\label{num:2}
		\end{align}	
		Hence when $N>81$, by computing $(N-81)\times(\ref{num:1})+(\ref{num:2})$, we have $R_{\textup{dec}}(M)\leq 2.00884 R$.

		To conclude, $R_{\textup{dec}}(M)\leq 2.00884 R$ holds for any $(M,R)\in \mathcal{S}_{\textup{Lower}}$ for all three cases. This completes the proof of Lemma \ref{lemma:decconst}.

  \section{Proof of Lemma \ref{lemma:ga}}\label{app:lemmaga}
			 	If $R$ is $\epsilon$-achievable, we can find message $X_{\boldsymbol{d}}$ such that for each user $k$, $W_{d_k}$ can be decoded from $Z_k$ and $X_{\boldsymbol{d}}$ with probability of error of at most $\epsilon$.		 	
			 	Using Fano's inequality, the following bound holds:
			 	\begin{align}\label{eq:fano}
			 	H(W_{d_k}|Z_k,X_{\boldsymbol{d}})\leq 1+\epsilon F\ \ \ \ \ \ \ \forall k\in\{1,...,K\}.
			 	\end{align}
			 	Equivalently, 
			 	\begin{align}
			 	H(X_{\boldsymbol{d}}|Z_k)\geq H(W_{d_k}&|Z_k)+H(X_{\boldsymbol{d}}|W_{d_k},Z_k) \nonumber\\&-(1+\epsilon F)\ \ \ \ \  \forall k\in\{1,...,K\}.
			 	\end{align}
			 	Note that the LHS of the above inequality lower bounds the communication load. If we lower bound the term $H(X_{\boldsymbol{d}}|W_{d_k},Z_k)$ on the RHS by $0$, we obtain the single user cutset bound. However, we enhance this cutset bound by bounding $H(X_{\boldsymbol{d}}|W_{d_k},Z_k)$ with non-negative functions. On a high level, we 			 	
			 	view $H(X_{\boldsymbol{d}}|W_{d_k},Z_k)$ as the communication load on an enhanced caching system, where $W_{d_k}$ and $Z_k$ are known by all the users. Using similar approach, we can lower bound $H(X_{\boldsymbol{d}}|W_{d_k},Z_k)$ by the sum of a single cutset bound on this enhanced system, and another entropy function that can be interpreted as the communication load on a further enhanced system. We can recursively apply this bounding technique until all user demands are publicly known.

			 	From (\ref{eq:fano}), we have
			 	\begin{align}
			 	H(W_{d_k}|Z_{\{1,...,k\}},X_{\boldsymbol{d}}, W_{\{d_1,...,d_{k-1}\}})&\leq 1+\epsilon F\nonumber\\&\ \forall k\in\{1,...,K\}.
			 	\end{align}
			 	Equivalently,
			 	\begin{align}
			 	H(X_{\boldsymbol{d}}|&Z_{\{1,...,k\}}, W_{\{d_1,...,d_{k-1}\}})\geq&\nonumber\\ &  H(W_{d_k}|Z_{\{1,...,k\}}, W_{\{d_1,...,d_{k-1}\}})\nonumber\\&+H(X_{\boldsymbol{d}}|Z_{\{1,...,k\}}, W_{\{d_1,...,d_{k}\}}) \nonumber\\   &-(1+\epsilon F)\ \ \ \ \ \ \ \ \ \ \ \ \ \ \ \ \ \ \ \  \forall k\in\{1,...,K\}.
			 	\end{align}
			 	Adding the above inequality for $k\in\{1,...,\min\{N,K\}\}$, we have
			 	\begin{align}
			 	H(X_{\boldsymbol{d}}&|Z_{\{1\}})
			 	\geq\nonumber\\ &\sum_{k=1}^{\min\{N,K\}} \left( H(W_{d_k}| Z_{\{1,...,k\}},W_{\{d_1,...,d_{k-1}\}})-(1+\epsilon F) \vphantom{H(W_{d_k}| Z_{\{1,...,k\}},W_{\{d_1,...,d_{k-1}\}})}\right)\nonumber\\&+H(X_{\boldsymbol{d}}|Z_{\{1,...,\min\{N,K\}\}}, W_{\{d_1,...,d_{\min\{N,K\}}\}})\nonumber\\
			 	\geq& \sum_{k=1}^{\min\{N,K\}} H(W_{d_k}| Z_{\{1,...,k\}},W_{\{d_1,...,d_{k-1}\}})\nonumber\\&-\min\{N,K\}(1+\epsilon F) .
			 	\end{align}
			 	Thus, $R$ is bounded by
			 		 	\begin{align}
			 		 	R\geq& \frac{1}{F} H(X_{\boldsymbol{d}}|Z_{\{1\}})\nonumber\\
			 		 	\geq& \frac{1}{F}(\sum_{k=1}^{\min\{N,K\}} H(W_{d_k}|Z_{\{1,...,k\}}, W_{\{d_1,...,d_{k-1}\}}))\nonumber\\&- \min\{N,K\}(\frac{1}{F}+\epsilon).
			 		 	\end{align}
  			 	One can show that this approach strictly improves the compound cutset bound, which was used in most of the prior works.

\section{Proof of Lemma \ref{lemma:stg}}\label{app:plstg}

In this appendix, we prove that for any prefetching scheme $\boldsymbol{\phi}$, the rate $R_{\textup{A}}(F,\boldsymbol{\phi})$ is lower bounded by the RHS of (\ref{ineq:p2l}), for any parameters $s$ and $\alpha$. Now we consider any such $s\in\{1,...,\min\{N,K\}\}$ and $\alpha \in [0,1]$. 
		 From the definition of $R_{\textup{A}}(F,\boldsymbol{\phi})$ and the non-negativity of entropy functions, we have
		 \begin{align}
		 R^*_{\textup{A}}(F,\boldsymbol{\phi})F\geq& \left
		 (\sum_{k=1}^{s-1} H^*(W_{k}|Z_{\{1,...,k\}}, W_{\{1,...,{k-1}\}})\right
		 )\nonumber\\&+\alpha H^*(W_{s}|Z_{\{1,...,s\}}, W_{\{1,...,{s-1}\}}). \label{bound:simple}
		 \end{align}
		 Each term in the above lower bound can be bounded in the following $2$ ways: \footnote{Rigorously, (\ref{bound:2}) requires $k<K$. However, we will only apply this bound for $k<s$, which satisfies this condition. }
		 \begin{align}
		H^*(W_{k}|Z_{\{1,...,k\}}&, W_{\{1,...,k-1\}}) \nonumber\\ \geq & \frac{H^*(W_{\{k,...,N\}}|Z_{\{1,...,k\}}, W_{\{1,...,{k-1}\}})}{N-k+1}\nonumber\\
		\geq &F-\frac{H^*(Z_{\{1,...,k\}}| W_{\{1,...,{k-1}\}})}{N-k+1}\label{bound:1}\\
		 H^*(W_{k}|Z_{\{1,...,k\}}&, W_{\{1,...,k-1\}})\nonumber\\= & F-H^*(Z_{\{1,...,k\}}| W_{\{1,...,{k-1}\}})\nonumber\\&+H^*(Z_{\{1,...,k\}}| W_{\{1,...,{k}\}})\nonumber\\
		 \geq &  F-H^*(Z_{\{1,...,k\}}| W_{\{1,...,{k-1}\}})\nonumber\\&+\frac{k}{k+1}H^*(Z_{\{1,...,k+1\}}| W_{\{1,...,{k}\}})\label{bound:2}
		 \end{align}
		 We aim to use linear combinations of the above two bounds in (\ref{bound:simple}), such that the coefficient of each $H^*(Z_{\{1,...,k\}}| W_{\{1,...,{k-1}\}})$ in the resulting lower bound is $0$ for all but one $k$.
		 To do so, we construct the following sequences:
		 	\begin{align}
		    a_x&=\frac{2\alpha s+s(s-1)-(x+1)x}{2x(N-x)},\\
		    b_x&=\frac{2\alpha s +s(s-1)-x(x-1)}{ 2x(N-x+1)}.
		\end{align} 
		We can verify that these sequences satisfy the following equations:
		 	\begin{align}
		    \frac{1-a_x}{N-x+1}+a_x=b_x,\\
		    \frac{x}{x+1} a_x=b_{x+1}.
		\end{align}
		Let $\ell\in\{1,...,s\}$ be the minimum value such that (\ref{eq:sml}) holds, we can prove that $a_x\in[0,1]$ for $x\in\{\ell,...,s-1\}$. Because $\ell$ is the minimum of such values, we can also prove that $b_l\geq \frac{\ell-1}{\ell}$. Using the above properties of sequences $\boldsymbol{a}$ and $\boldsymbol{b}$, we lower bound $R_{\textup{A}}(F,\boldsymbol{\phi})$ as follows:
		
	    For each $x\in\{\ell,...,s-1\}$, by computing $(1-a_x)\times(\ref{bound:1})+a_x\times(\ref{bound:2})$, 
		    we have
		    
	    	 \begin{align}
		H^*(W_{x}|&Z_{\{1,...,x\}}, W_{\{1,...,x-1\}})
		\nonumber\\\geq& (1-a_x)\left(F-\frac{H^*(Z_{\{1,...,k\}}| W_{\{1,...,{k-1}\}})}{N-k+1}\right)\nonumber\\&+a_x(F-H^*(Z_{\{1,...,k\}}| W_{\{1,...,{k-1}\}})\nonumber\\&+\frac{k}{k+1}H^*(Z_{\{1,...,k+1\}}| W_{\{1,...,{k}\}}))\nonumber\\
		=&F- (\frac{1-a_x}{N-x+1}+a_x) H^*(Z_{\{1,...,x\}}| W_{\{1,...,{x-1}\}})\nonumber\\&+a_x\frac{x}{x+1}H^*(Z_{\{1,...,x+1\}}| W_{\{1,...,{x}\}} )\nonumber\\
		=&F- b_x H^*(Z_{\{1,...,x\}}| W_{\{1,...,{x-1}\}})\nonumber\\&+b_{x+1}H^*(Z_{\{1,...,x+1\}}| W_{\{1,...,{x}\}} ).
		\end{align}
		Moreover, we have the follows from (\ref{bound:1}):
		\begin{align}
		\alpha H^*(W_{s}|&Z_{\{1,...,s\}}, W_{\{1,...,s-1\}})
	\nonumber\\&\geq \alpha\left( F-\frac{H^*(Z_{\{1,...,s\}}| W_{\{1,...,{s-1}\}})}{N-s+1}\right)\nonumber\\
		&=\alpha F -b_s H^*(Z_{\{1,...,s\}}| W_{\{1,...,{s-1}\}})  .
		 \end{align}
		Consequently,
		\begin{align}
		\sum_{k=\ell}^{s-1}  H^*(&W_{k}|Z_{\{1,...,k\}}, W_{\{1,...,{k-1}\}})\nonumber\\+\alpha& H^*(W_{s}|Z_{\{1,...,s\}}, W_{\{1,...,s-1\}}) \geq \nonumber\\& (s-\ell+\alpha)F-b_\ell H^*(Z_{\{1,...,\ell\}}| W_{\{1,...,{\ell-1}\}}).\label{eq:ltos}
		\end{align}
		On the other hand, 
		\begin{align}
		\sum_{k=1}^{\ell-1}  H^*(W_{k}|&Z_{\{1,...,k\}}, W_{\{1,...,{k-1}\}})\nonumber\\\geq& \sum_{k=1}^{\ell-1} ( F-H^*(Z_{\{1,...,k\}}| W_{\{1,...,{k-1}\}})\nonumber\\&+\frac{k}{k+1}H^*(Z_{\{1,...,k+1\}}| W_{\{1,...,{k}\}}))\nonumber\\
		=&\sum_{k=1}^{\ell-1} ( F-\frac{1}{k}H^*(Z_{\{1,...,k\}}| W_{\{1,...,{k-1}\}}))\nonumber\\&+\frac{\ell-1}{\ell}H^*(Z_{\{1,...,\ell\}}| W_{\{1,...,{\ell-1}\}})\nonumber\\
		\geq& (\ell-1) F-(\ell-1)MF \nonumber\\&+\frac{\ell-1}{\ell}H^*(Z_{\{1,...,\ell\}}| W_{\{1,...,{\ell-1}\}}).\label{eq:1tol}
		\end{align}
		Combining (\ref{bound:simple}), (\ref{eq:ltos}), and (\ref{eq:1tol}), we have
			\begin{align}
				R_{\textup{A}}(F,\boldsymbol{\phi})F
				\geq&  (\ell-1) F-(\ell-1)MF \nonumber\\&+\frac{\ell-1}{\ell}H^*(Z_{\{1,...,\ell\}}| W_{\{1,...,{\ell-1}\}})\nonumber\\&+(s-\ell+\alpha)F-b_\ell H^*(Z_{\{1,...,\ell\}}| W_{\{1,...,{\ell-1}\}})\nonumber\\
				=& (s-1+\alpha)F-\left(\ell-1\right)MF \nonumber\\&+\left(\frac{\ell-1}{\ell}-b_l\right)H^*(Z_{\{1,...,\ell\}}| W_{\{1,...,{\ell-1}\}}) .
				\end{align}
				Recall that $b_\ell\geq \frac{\ell-1}{\ell}$, we have
						\begin{align}
							R_{\textup{A}}(F,\boldsymbol{\phi})F
						\geq& (s-1+\alpha)F- (\ell-1)MF\nonumber\\&- (b_\ell-\frac{\ell-1}{\ell})\ell MF\nonumber\\
						=&(s-1+\alpha)F\nonumber\\&- \frac{s(s-1)-\ell(\ell-1)+2\alpha s}{2(N-\ell+1)}MF.
						\end{align}
						This completes the proof of Lemma \ref{lemma:stg}.

	\section{Proof of Lemma \ref{lemma:3}}\label{app:lm3}

         		To simplify the discussion, we adopt the notation of $H^*(W_\mathcal{A}, Z_{\mathcal{B}})$ which is defined in the proof of Theorem \ref{th:general}. Moreover, we generalize this notation to include the variables for the messages $X_{\boldsymbol{d}}$. For any permutations $p\in\mathcal{P}_N$, $q\in\mathcal{P}_K$ and for any demand $\boldsymbol{d}\in \{1,...,N\}^K$, we define $\boldsymbol{d}(p,q)$ be a demand where for each $k\in\{1,...,K\}$, user $q(k)$ requests file $p(d_k)$.
			    Then for any subset for demands $\mathcal{D}\subseteq \{1,...,N\}^K$, we define $\mathcal{D}(p,q)=\{\boldsymbol{d}{(p,q)}|\boldsymbol{d}\in \mathcal{D}\}$.
			    Now for any subsets	$\mathcal{A}\subseteq \{1,...,N\}$, $\mathcal{B}\subseteq \{1,...,K\}$, and $\mathcal{D}\subseteq \{1,...,N\}^K$, we define
				\begin{align}
			    H^*(X_{\mathcal{D}}, &W_{\mathcal{A}},Z_{\mathcal{B}})\nonumber\\&\triangleq \frac{1}{N!K!} \sum_{p\in \mathcal{P}_N, q\in \mathcal{P}_K} H(X_{\mathcal{D}(p,q)}, W_{p\mathcal{A}},Z_{q\mathcal{B}}).\label{eq:star}
			\end{align}
					
					For any $i\in\{1,..,n\}$ and $j\in\{1,...,\alpha\}$ let $\boldsymbol{d}^{i,j}$ be a demand satisfying
					\begin{align}
					    d^{i,j}_l=\begin{cases}
					     l-i+ (j-1)(K-n) +\beta \\\ \ \ \ \ \ \ \ \ \ \ \ \ \ \ \ \ \ \ \ \ \ \ \ \ \ \ \ \text{if } i+1\leq l\leq i+K-n,\\
					     1 \ \ \ \ \ \ \ \ \ \ \ \ \ \ \ \ \ \ \ \ \ \ \ \ \ \ \text{otherwise}.
					    \end{cases}
					\end{align}
					Note that for all demands $\boldsymbol{d}^{i,j}$, user $1$ requests file $1$, hence we have
					\begin{align}
					   H(W_1|X_{\boldsymbol{d}^{i,j}}, Z_1)\leq 1+\epsilon F
					\end{align}
					using Fano's inequality.
					Consequently,
					\begin{align}
					   RF&\geq H(X_{\boldsymbol{d}^{i,j}}) \nonumber\\
					  &\geq H(X_{\boldsymbol{d}^{i,j}}|Z_1)+H(W_1|X_{\boldsymbol{d}^{i,j}}, Z_1)-(1+\epsilon F)\nonumber\\
					   &=H(W_1| Z_1)+ H(X_{\boldsymbol{d}^{i,j}} |W_1 ,Z_1)-(1+\epsilon F).
					\end{align}
					Due to the homogeneity of the problem, we have
						\begin{align}
					   RF&\geq H^*(W_1| Z_1)+ H^*(X_{\boldsymbol{d}^{i,j}}|W_1 ,Z_{1})-(1+\epsilon F).
					\end{align}
					For each $i\in\{1,...,n\}$, $j\in\{1,...,\alpha\}$, and $k\in\{1,...,i\}$, we have the following identity:
						\begin{align}
					 H^*(X_{\boldsymbol{d}^{i,j}}|W_1 ,Z_1)= H^*(X_{\boldsymbol{d}^{i,j}}|W_1 ,Z_k).
					\end{align}
					Hence, we have
							\begin{align}
					   RF\geq &H^*(W_1| Z_1)\nonumber\\&+\frac{2}{n(n+1)\alpha}  \sum_{k=1}^{n} \sum_{i=k}^n \sum_{j=1}^\alpha H^*(X_{\boldsymbol{d}^{i,j}}|W_1 ,Z_k)\nonumber\\& -(1+\epsilon F).
					\end{align}
					
					For $k\in\{1,...,n\}$, let $\mathcal{D}_{k}$ and $\mathcal{D}_{k}^+$ denote the following set of demands:
					\begin{align}
					    \mathcal{D}_k&=\{\boldsymbol{d}^{k,j}|  j\in\{1,...,\alpha\}
					    \},\\
					    \mathcal{D}_k^+&=\bigcup_{i=k}^{n}\mathcal{D}_{i},
					\end{align}
				we have
						\begin{align}
					   RF\geq& H^*(W_1| Z_1)+\frac{2}{n(n+1)\alpha}  \sum_{k=1}^{n}  H^*(X_{\mathcal{D}_k^+}| W_1 ,Z_k)\nonumber\\&-(1+\epsilon F)\nonumber\\
					   \geq& H^*(W_1| Z_1)+\frac{2}{n(n+1)\alpha}  \sum_{k=1}^{n}  H^*(X_{\mathcal{D}_k^+}| W_{\{1,...,\beta\}} ,Z_k)\nonumber\\&-(1+\epsilon F)\nonumber\\
				\geq&  H^*(W_1| Z_1)\nonumber\\&+\frac{2}{n(n+1)\alpha} \sum_{k=1}^n\left( H^*( Z_k, X_{\mathcal{D}_k^+}|W_{\{1,...,\beta\}} )\right.\nonumber\\&\left.\vphantom{H^*( Z_k, X_{\mathcal{D}_k^+}|W_{\{1,...,\beta\}})}-H^*(Z_k|W_{\{1,...,\beta\}})\right)-(1+\epsilon F).\label{eq:111}
					\end{align}
					To further bound $R$, we only need a lower bound for $\sum\limits_{k=1}^n H^*( Z_k, X_{\mathcal{D}_k^+}|W_{\{1,...,\beta\}} )$, which is derived as follows:
					
					For each $i\in \{1,...,K-n\}$, let $\mathcal{S}_i$ be subset of files defined as follows:
					\begin{align}
					    \mathcal{S}_i&=\{i+(j-1)(K-n)+\beta\ |\  j\in\{1,...,\alpha\}\}.
					\end{align}
				 From the decodability constraint, for any $k\in\{1,...,n\}$, each file in $\mathcal{S}_i$ can be decoded by user $i+k$ given $X_{\mathcal{D}_{k}}$. Using Fano's inequality, we have
						\begin{align}
					   H^*(W_{\mathcal{S}_i}|X_{\mathcal{D}_{k}}, Z_{i+k})\leq \alpha(1+\epsilon F).
					\end{align}
						Let $\mathcal{S}_i^-$ be subset of files defined as follows
						\begin{align}
						    \mathcal{S}_i^-=  \left(\bigcup_{j=1}^{i} \mathcal{S}_j\right )\  \bigcup \ \{1,...,\beta\}.
						\end{align}
						We have
					\begin{align}
					   0\geq& H^*(W_{\mathcal{S}_i}|X_{\mathcal{D}_{k}^+}, Z_{i+k}, W_{\mathcal{S}_{i-1}^-})-\alpha(1+\epsilon F)\nonumber\\
					   =&H^*(X_{\mathcal{D}_{k}^+}, Z_{i+k}|  W_{\mathcal{S}_{i}}, W_{\mathcal{S}_{i-1}^-})+H^*( W_{\mathcal{S}_i}|  W_{\mathcal{S}_{i-1}^-})\nonumber\\&-H^*(X_{\mathcal{D}_{k}^+}, Z_{i+k}|  W_{\mathcal{S}_{i-1}^-})-\alpha(1+\epsilon F)\nonumber\\
					   =&H^*(X_{\mathcal{D}_{k}^+}, Z_{i+k}|   W_{\mathcal{S}_{i}^-})+\alpha F-H^*(X_{\mathcal{D}_{k}^+}, Z_{i+k}|  W_{\mathcal{S}_{i-1}^-})\nonumber\\&-\alpha(1+\epsilon F) .
					\end{align}
					Consequently, 
					\begin{align}
					   0  \geq&\sum_{k=1}^{n} \sum_{i=1}^{K-n} \left(H^*(X_{\mathcal{D}_{k}^+}, Z_{i+k}|   W_{\mathcal{S}_{i}^-})+\alpha F\right.\nonumber\\&\left.-H^*(X_{\mathcal{D}_{k}^+}, Z_{i+k}|  W_{\mathcal{S}_{i-1}^-}) -\alpha (1+\epsilon F)\right)\nonumber\\
					   =&\sum_{k=1}^{n} \left(\sum_{i=1}^{K-n} \left(H^*(X_{\mathcal{D}_{k}^+}, Z_{i+k-1}|   W_{\mathcal{S}_{i-1}^-})\right.\right.\nonumber\\&\left.-H^*(X_{\mathcal{D}_{k}^+}, Z_{i+k}|  W_{\mathcal{S}_{i-1}^-}) \right)
					     +H^*(X_{\mathcal{D}_{k}^+}, Z_{K-n+k}|   W_{\mathcal{S}_{n}^-}) \nonumber\\&\left.-H^*(X_{\mathcal{D}_{k}^+}, Z_{k}|   W_{\mathcal{S}_{0}^-})\vphantom{\sum_{i=1}^{K-n}}\right)\nonumber\\&\ +\alpha n(K-n)(F-1-\epsilon F)\nonumber\\
					    \geq& \sum_{k=1}^{n} \left(\sum_{i=1}^{K-n} \left(H^*(X_{\mathcal{D}_{k}^+}, Z_{i+k-1}|   W_{\mathcal{S}_{i-1}^-})\right.\right.\nonumber\\&\left.\left.-H^*(X_{\mathcal{D}_{k}^+}, Z_{i+k}|  W_{\mathcal{S}_{i-1}^-}) \right)
					     -H^*(X_{\mathcal{D}_{k}^+}, Z_{k}|   W_{\mathcal{S}_{0}^-})\vphantom{\sum_{i=1}^{K-n}}\right)\nonumber\\&+\alpha n(K-n)(F-1-\epsilon F).
					\end{align}
					Hence, we obtain the following lower bound:
						\begin{align}
					 \sum_{k=1}^{n} H^*(&X_{\mathcal{D}_{k}^+}, Z_{k}|   W_{\mathcal{S}_{0}^-}))\nonumber\\
					 \geq&  \sum_{k=1}^{n} \sum_{i=1}^{K-n} \left(H^*(X_{\mathcal{D}_{k}^+}, Z_{i+k-1}|   W_{\mathcal{S}_{i-1}^-})\right.\nonumber\\&\left.-H^*(X_{\mathcal{D}_{k}^+}, Z_{i+k}|  W_{\mathcal{S}_{i-1}^-}) \right)\nonumber\\&+\alpha n(K-n)(F-1-\epsilon F)\nonumber\\
					 =& \sum_{i=1}^{K-n}\sum_{k=1}^{n} \left(H^*(X_{\mathcal{D}_{k}^+}, Z_{i+k-1}|   W_{\mathcal{S}_{i-1}^-})\right.\nonumber\\&\left.-H^*(X_{\mathcal{D}_{k}^+}, Z_{i+k}|  W_{\mathcal{S}_{i-1}^-}) \right)\nonumber\\&+\alpha n(K-n)(F-1-\epsilon F)\nonumber\\
					 =& \sum_{i=1}^{K-n}\sum_{k=1}^{n} \left(H^*( Z_{i+k-1}|X_{\mathcal{D}_{k}^+},   W_{\mathcal{S}_{i-1}^-})\right.\nonumber\\&\left.-H^*( Z_{i+k}|X_{\mathcal{D}_{k}^+},  W_{\mathcal{S}_{i-1}^-}) \right)\nonumber\\&+\alpha n(K-n)(F-1-\epsilon F)
					 \end{align}
					 Note that $\mathcal{D}_{k}^+\subseteq \mathcal{D}_{k-1}^+$, we have $H^*( Z_{i+k}|X_{\mathcal{D}_{k+1}^+},  W_{\mathcal{S}_{i-1}^-}) \geq H^*( Z_{i+k}|X_{\mathcal{D}_{k}^+},  W_{\mathcal{S}_{i-1}^-})$. Consequently,
						\begin{align}
					 \sum_{k=1}^{n} H&^*(X_{\mathcal{D}_{k}^+}, Z_{k}|   W_{\mathcal{S}_{0}^-}))\nonumber\\
					 \geq& \sum_{i=1}^{K-n} \left(H^*( Z_{i}|X_{\mathcal{D}_{1}^+},   W_{\mathcal{S}_{i-1}^-})-H^*( Z_{i+n}|X_{\mathcal{D}_{n}^+},  W_{\mathcal{S}_{i-1}^-}) \right)\nonumber\\&+\alpha n(K-n)(F-1-\epsilon F)\nonumber\\
					 \geq& -\sum_{i=1}^{K-n} H^*( Z_{i+n}|  W_{\mathcal{S}_{i-1}^-} )+\alpha n(K-n)(F-1-\epsilon F)\nonumber\\
					 =&-\sum_{i=1}^{K-n} H^*( Z_{1}|  W_{\{1,...,\beta+i\alpha\}} )\nonumber\\&+\alpha n(K-n)(F-1-\epsilon F).\label{eq:130}
					 \end{align}
    				Applying (\ref{eq:130}) to  (\ref{eq:111}), we have
					\begin{align}
					    RF\geq& H^*(W_1| Z_1)+\frac{2}{n(n+1)\alpha} \left(\vphantom{\sum_{i=0}^{K-n-1}}\alpha n(K-n)(F-1-\epsilon F)\right.
					  \nonumber\\&-\sum_{k=1}^{n} H^*(Z_k|W_{\{1,...,\beta\}})\nonumber\\&\left.-\sum_{i=0}^{K-n-1} H^*(Z_{1}|  W_{\{1,...,\beta+i\alpha\}})\right)\nonumber\\
					  &-(1+\epsilon F)\nonumber\\
					  =& H^*(W_1| Z_1)+\frac{2}{n(n+1)\alpha} \left(\vphantom{\sum_{i=0}^{K-n-1}}\alpha n(K-n)F\right.
					  \nonumber\\&- nH^*(Z_1|W_{\{1,...,\beta\}})\nonumber\\&\left.-\sum_{i=0}^{K-n-1} H^*(Z_{1}|  W_{\{1,...,\beta+i\alpha\}})\right)\nonumber\\
					  & -\frac{2K-n+1}{n+1}(1+\epsilon F).
					\end{align}

		\section{Proof of Theorem \ref{th:gconst} for average rate} \label{app:ave1}
		
		Here we prove Theorem \ref{th:gconst} for the average rate (i.e. inequalities (\ref{eq:ruave}) and (\ref{eq:lruave})). The upper bounds of $R^*_{\textup{ave}}$ in these inequalities can be achieved using the caching scheme provided in \cite{yu2016exact}, hence we only need to prove their lower bounds.
	To do so, we define the following terminology:
	
	We divide the set of all demands, denoted by $\mathcal{D}$, into smaller subsets, and refer them to as \emph{types}. We use the same definition in \cite{yu2016exact}, which are stated as follows:
		 Given an arbitrary demand $\boldsymbol{d}$, we define its \textit{statistics}, denoted by  $\boldsymbol{s}(\boldsymbol{d})$, as a sorted array of length $N$, such that $s_i(\boldsymbol{d})$ equals the number of users that request the $i$th most requested file. We denote the set of all possible statistics by $\mathcal{S}$. 
	 Grouping by the same statistics, the set of all demands $\mathcal{D}$ can be broken into many subsets.	For any statistics $\boldsymbol{s}\in\mathcal{S}$, we define type $\mathcal{D}_{\boldsymbol{s}}$ as the set of queries with statistics $\boldsymbol{s}$. 
		 Note that for each demand  $\boldsymbol{d}$, the value $N_{\textup{e}}(\boldsymbol{d})$ only depends on its statistics $\boldsymbol{s}({\boldsymbol{d}})$, and thus the value is identical across all demands in $\mathcal{D}_{\boldsymbol{s}}$. For convenience, we denote that value by $N_{\textup{e}}(\boldsymbol{s})$.
		 
		 Given a prefetching scheme $\boldsymbol{\phi}$ and a type $\mathcal{D}_{\boldsymbol{s}}$, we say a rate $R$ is \textit{$\epsilon$-{achievable for type}} $\mathcal{D}_{\boldsymbol{s}}$ if
		 we can find a function $R(\boldsymbol{d})$ that  
	is $\epsilon$-achievable for any demand $\boldsymbol{d}$ in $\mathcal{D}_{\boldsymbol{s}}$, satisfying $R=\mathbb{E}_{\boldsymbol{d}}[R(\boldsymbol{d})]$, where $\boldsymbol{d}$ is uniformly random in $\mathcal{D}_{\boldsymbol{s}}$. Hence, to characterize $R^*_{\textup{ave}}$, it is sufficient to lower bound the {$\epsilon$-{achievable}} rates for each type individually, and show that for each type, the caching scheme provided in \cite{yu2016exact} is within the given constant factors optimal for large $F$ and small $\epsilon$. 
	
	We first lower bound any $\epsilon$-achievable rate for each type as follows 
	: Within a type $\mathcal{D}_{\boldsymbol{s}}$, we can find a demand $\boldsymbol{d}$, such that users in $\{1,...,N_{\textup{e}}(\boldsymbol{s})\}$ requests different files. We can easily generalize Lemma \ref{lemma:ga} to this demand, and any $\epsilon$ achievable rate of this demand, denoted by $R_{\boldsymbol{d}}$, is lower bounded by the following inequality: 
				\begin{align}
				R_{\boldsymbol{d}}\geq& \frac{1}{F}\left(\sum_{k=1}^{N_{\textup{e}}(\boldsymbol{s})} H(W_{d_k}|Z_{\{1,...,k\}}, W_{\{d_1,...,d_{k-1}\}})\right)\nonumber\\&- N_{\textup{e}}(\boldsymbol{s})(\frac{1}{F}+\epsilon). 
				\end{align}	
	Applying the same bounding technique to all demands in type $\mathcal{D}_{\boldsymbol{s}}$. We can prove that any rate that is $\epsilon$-achievable for $\mathcal{D}_{\boldsymbol{s}}$, denoted by $R_{\boldsymbol{s}}$, is bounded by the follows:
	\begin{align}
				R_{\boldsymbol{s}}\geq& \frac{1}{F}\left(\sum_{k=1}^{N_{\textup{e}}(\boldsymbol{s})} H^*(W_{k}|Z_{\{1,...,k\}}, W_{\{1,...,{k-1}\}})\right)\nonumber\\&- N_{\textup{e}}(\boldsymbol{s})(\frac{1}{F}+\epsilon),\label{eq:93}
				\end{align}	
				where function $H^*(\cdot)$ is defined in the proof of Theorem \ref{th:general}.
				
	Following the same steps in the proof of Theorem  \ref{th:general}, we can prove that 
			\begin{equation}
			R_{\boldsymbol{s}}\geq s-1+\alpha-\frac{s(s-1)-\ell(\ell-1)+ 2\alpha s}{2(N-\ell+1)} M - N_{\textup{e}}(\boldsymbol{s})(\frac{1}{F}+\epsilon), \label{eq:94}
			\end{equation}
			for arbitrary $s\in\{1,...,N_{\textup{e}}(\boldsymbol{s})\}$, $\alpha\in[0,1]$, where $\ell\in\{1,...,s\}$ is the minimum value such that 
			\begin{equation}
			\frac{s(s-1)-\ell (\ell-1)}{2}+\alpha s\leq  (N-\ell+1)\ell.
			\end{equation}

			On the other hand, the caching scheme provided in \cite{yu2016exact} achieves an average rate of $\textup{Conv}\left(
		\frac{\binom{K}{r+1}-\binom{K-N_{\textup{e}}(\boldsymbol{s})}{r+1}}{\binom{K}{r}}\right)$ within each type $\mathcal{D}_{\boldsymbol{s}}$. Using the results in \cite{yu2016exact}, we can easily prove that this average rate can be upper bounded by $R_{\textup{dec}}(M,\boldsymbol{s})$, defined as
		\begin{align}
		    R_{\textup{dec}}(M,\boldsymbol{s})\triangleq\frac{N-M}{M}(1-(1-\frac{M}{N})^{N_{\textup{e}}(\boldsymbol{s})}).\label{eq:dectype}
		\end{align}
		Hence, in order to prove (\ref{eq:ruave}) and (\ref{eq:lruave}), it suffices to prove that for large $F$ and small $\epsilon$, any $\epsilon$-achievable rate $R_{\boldsymbol{s}}$ for any type $\mathcal{D}_{\boldsymbol{s}}$ satisfies $R_{\boldsymbol{s}}\geq R_{\textup{dec}}(M,\boldsymbol{s})/2.00884 $ in the general case, and  $R_{\boldsymbol{s}}\geq R_{\textup{dec}}(M,\boldsymbol{s})/2 $ when $N\geq \frac{K(K+1)}{2}$.
		
		Note that the above characterization of $R_{\boldsymbol{s}}$ exactly matches a characterization of $R^*$ for a caching system with $N$ files and $N_{\textup{e}}(\boldsymbol{s})$ users. Specifically, the lower bound of $R_{\boldsymbol{s}}$ given by (\ref{eq:94}) exactly matches Theorem \ref{th:general}, and the upper bound $R_{\textup{dec}}(M,\boldsymbol{s})$ defined in (\ref{eq:dectype}) exactly matches the upper bound $R_
		{\textup{dec}}(M)$ defined in (\ref{eq:decdef}). Thus, by reusing the same arguments in the proof of Theorem \ref{th:gconst} for the peak rate, we can easily prove that $R_{\boldsymbol{s}}\geq R_{\textup{dec}}(M,\boldsymbol{s})/2.00884 $ holds for the general case, and $R_{\boldsymbol{s}}\geq R_{\textup{dec}}(M,\boldsymbol{s})/2 $ holds for sufficiently large $N$ when $\frac{N_{\textup{e}}(\boldsymbol{s})M}{N}>1$. Hence, to prove Theorem \ref{th:gconst} for the average rate, we only need $R_{\boldsymbol{s}}\geq R_{\textup{dec}}(M,\boldsymbol{s})/2$ for sufficiently large $N$ to also hold when $\frac{N_{\textup{e}}(\boldsymbol{s})M}{N}\leq 1$, which can be easily proved as follows:
		
		Using the same arguments in the proof of Theorem \ref{th:gconst} for the peak rate, the following inequality can be derived from (\ref{eq:94}) for large $N$, large $F$ and small $\epsilon$:
			\begin{equation}
			R_{\boldsymbol{s}}\geq N_{\textup{e}}(\boldsymbol{s})-\frac{N_{\textup{e}}(\boldsymbol{s})(N_{\textup{e}}(\boldsymbol{s})+1)}{2} \cdot \frac{M}{N},
			\end{equation}
		which is a linear function of $M$. Furthermore, since $R_{\textup{dec}}(M,\boldsymbol{s})$ is convex, we only need to check that
				\begin{align}
		    \frac{R_{\textup{dec}}(M,\boldsymbol{s})}{2}\leq N_{\textup{e}}(\boldsymbol{s})-\frac{N_{\textup{e}}(\boldsymbol{s})(N_{\textup{e}}(\boldsymbol{s})+1)}{2} \cdot \frac{M}{N}
		\end{align}
		holds at $\frac{N_{\textup{e}}(\boldsymbol{s})M}{N}\in \{0, 1\}$.
		
		\noindent For $\frac{N_{\textup{e}}(\boldsymbol{s})M}{N}=0$, we have
			\begin{align}
		    \frac{R_{\textup{dec}}(M,\boldsymbol{s})}{2}=& \frac{N_{\textup{e}}(\boldsymbol{s})}{2}\leq  N_{\textup{e}}(\boldsymbol{s})= \nonumber\\&N_{\textup{e}}(\boldsymbol{s})-\frac{N_{\textup{e}}(\boldsymbol{s})(N_{\textup{e}}(\boldsymbol{s})+1)}{2} \cdot \frac{M}{N}.
		\end{align}
		For $\frac{N_{\textup{e}}(\boldsymbol{s})M}{N}=1$, we have
		\begin{align}
		    \frac{R_{\textup{dec}}(M,\boldsymbol{s})}{2}=& \frac{N_{\textup{e}}(\boldsymbol{s})-1}{2}\left(1-\left(1-\frac{1}{N_{\textup{e}}(\boldsymbol{s})}\right)^{N_{\textup{e}}(\boldsymbol{s})}\right)\nonumber\\\leq&  \frac{N_{\textup{e}}(\boldsymbol{s})-1}{2}\nonumber\\=& N_{\textup{e}}(\boldsymbol{s})-\frac{N_{\textup{e}}(\boldsymbol{s})(N_{\textup{e}}(\boldsymbol{s})+1)}{2} \cdot \frac{M}{N}.
		\end{align}
			This completes the proof of Theorem \ref{th:gconst}.

 	\section{The exact rate-memory tradeoff for two-user case} \label{app:ave2}
 	As mentioned in Remark \ref{rem:twoave}, we can completely characterize the rate-memory tradeoff for average rate for the two-user case, for any possible values of $N$ and $M$. We formally state this result in the following corollary:
			\begin{corollary}\label{cor:twoave}
			For a caching system with $2$ users, a database of $N$ files, and a local cache size of $M$ files at each user, we have
			\begin{equation}
		 R^*_{\textup{ave}}=R_{\textup{u,ave}}(N,K,r),
			\end{equation}
			where  $R_\textup{u,ave}(N,K,r)$ is defined in Definition  \ref{def}.
	\end{corollary}
		\begin{proof}
		For the single-file case, only one possible demand exists. The average rate thus equals the peak rate, which can be easily characterized. Hence, we omit the proof and focus on cases where $N\geq 2$. Note that $R_{\textup{u,ave}}$ can be achieved using the scheme provided in \cite{yu2016exact}, we only need to prove that $R^*_{\textup{ave}}\geq R_{\textup{u,ave}}(N,K,r)$.
		
		As shown in Appendix \ref{app:ave1}, the average rate within each type $\mathcal{D}_{\boldsymbol{s}}$ is bounded by (\ref{eq:93}). Hence, the minimum average rate under uniform file popularity given a prefetching scheme $\phi$, denoted by $R(\boldsymbol{\phi})$, is lower bounded by
		\begin{align}
			R(\boldsymbol{\phi})\geq& \mathbb{E}_{\boldsymbol{s}}\left[\frac{1}{F}\left(\sum_{k=1}^{N_{\textup{e}}(\boldsymbol{s})} H^*(W_{k}|Z_{\{1,...,k\}}, W_{\{1,...,{k-1}\}})\right)\right.\nonumber\\&\left.- N_{\textup{e}}(\boldsymbol{s})(\frac{1}{F}+\epsilon)\vphantom{\sum_{k=1}^{N_{\textup{e}}(\boldsymbol{s})}}\right].
				\end{align}	
		Note that for the two-user case, $N_{\textup{e}}(\boldsymbol{s})$ equals $1$ with probability $\frac{1}{N}$, and $2$ with probability $\frac{N-1}{N}$. Consequently,
		\begin{align}
			R(\boldsymbol{\phi})\geq& \frac{1}{F}\left( H^*(W_{1}|Z_{1})+\frac{N-1}{N}\cdot H^*(W_{2}|Z_{\{1,2\}}, W_{1})\right)\nonumber\\&- \frac{2N-1}{N}\cdot (\frac{1}{F}+\epsilon).
				\end{align}	
		Using the technique developed in proof of Theorem \ref{th:general}, we have the following two lower bounds
		\begin{align}
			R(\boldsymbol{\phi})\geq& \frac{1}{F}H^*(W_{1}|Z_{1})- \frac{2N-1}{N}\cdot (\frac{1}{F}+\epsilon)\nonumber\\
			\geq& 1-\frac{M}{N}- \frac{2N-1}{N}\cdot (\frac{1}{F}+\epsilon),\\
			R(\boldsymbol{\phi})\geq& \frac{1}{F}\left( H^*(W_{1}|Z_{1})+\frac{N-1}{N}\cdot H^*(W_{2}|Z_{\{1,2\}}, W_{1})\right)\nonumber\\&- \frac{2N-1}{N}\cdot (\frac{1}{F}+\epsilon)\nonumber\\
			\geq& \frac{1}{F}\left( H^*(W_{1}|Z_{1})+\frac{1}{N}\cdot ((N-1)F-2 H^*(Z_{1}| W_{1}))\right)\nonumber\\&- \frac{2N-1}{N}\cdot (\frac{1}{F}+\epsilon)\nonumber\\
			\geq& \frac{2N-1}{N}-\frac{3N-2}{N}\cdot \frac{M}{N}-\frac{2N-1}{N}\cdot (\frac{1}{F}+\epsilon).
				\end{align}
		Hence we have
		\begin{align}
			R^*_{\textup{ave}}\geq& \max\left\{1-\frac{M}{N},\  \frac{2N-1}{N}-\frac{3N-2}{N}\cdot \frac{M}{N} \right\}\nonumber\\=&R_{\textup{u,ave}}(N,K,r).
				\end{align}
		
	 \end{proof}

		\section{Proof of Theorem \ref{th:exact5} for average rate} \label{app:ave3}
		To prove Theorem \ref{th:exact5} for the average rate, we need to show that $R^*_{\textup{ave}}=R_{\textup{u}}(N,K,r)$ for large $N$, for any caching system with no more than $5$ users.  Note that when $N$ is large, with high probability all users will request distinct files. Hence,
	 we only need to prove that the minimum average rate within the \emph{type of the worst case demands} (i.e., the set of demands where all users request distinct files) equals $R_{\textup{u}}(N,K,r)$.
	 Since $R_{\textup{u}}(N,K,r)$ can already be achieved according to \cite{yu2016exact}, it suffices to prove that this average rate is lower bounded by $R_{\textup{u}}(N,K,r)$.

	 Similar to the peak rate case, we prove that this fact holds if $\frac{KM}{N}\leq 1$ or $\frac{KM}{N}\geq \lceil \frac{K-3}{2}\rceil$ for large $N$. 
	 When $\frac{KM}{N}\leq 1$ or $\frac{KM}{N}\geq K-1$, this can be proved the same way as Lemma \ref{lemma:2}, while for the other case (i.e. $\frac{KM}{N}\in \left[\max\{\lceil\frac{K-3}{2}\rceil,1\}, K-1\right)$), we need to prove a new version of Theorem \ref{th:2}, which lower bounds the average rate within the type of the worst case demands.  
						To simplify the discussion, we adopt the notation of $ H^*(X_{\mathcal{D}}, W_{\mathcal{A}},Z_{\mathcal{B}})$ which is defined in (\ref{eq:star}). 
			We  also adopt the corresponding notation for conditional entropy.
					Suppose rate $R$ is $\epsilon$ achievable for the worst case type, we start by proving converse bounds of $R$ for large $N$.
					
					Recall that $r=\frac{KM}{N}$, and let $n=\lfloor r+1 \rfloor$.
					Because $r\in \left[1, K-1\right)$, we have $n\in\{2,...,K-1\}$. Let 
					$\alpha=\lfloor \frac{N-K}{K-n} \rfloor$ and $\beta=N-\alpha(K-n)$. Suppose $N$ is large enough,  such that $\alpha>0$. For any					
					$i\in\{1,..,n\}$ and $j\in\{1,...,\alpha\}$ let $\boldsymbol{d}^{i,j}$ be a demand satisfying
					\begin{align}
					    d^{i,j}_l&=\begin{cases}
					     l-i+ (j-1)(K-n) +\beta \\\ \ \ \ \ \ \ \ \ \ \ \ \ \ \ \ \ \ \ \ \ \ \ \ \ \ \ \ \text{if } i+1\leq l\leq i+K-n,\\
					     l\ \ \ \ \ \ \ \ \ \ \ \ \ \ \ \ \ \ \ \ \ \ \ \ \ \ \  \text{otherwise}.
					    \end{cases}
					\end{align}
					Note that the above demands belong to the worst case type, so we have $ RF\geq H^*(X_{\boldsymbol{d}^{i,j}}) $ for any $i$ and $j$. Following the same steps of proving Lemma \ref{lemma:3}, we have
					\begin{align}
					    RF\geq&  H^*(W_1| Z_1)+\frac{2}{n(n+1)\alpha} \left(\vphantom{\sum_{i=0}^{K-n-1}}\alpha n(K-n)F\right.
					  \nonumber\\& \left.- nH^*(Z_1|W_{\{1,...,\beta\}})-\sum_{i=0}^{K-n-1} H^*(Z_{1}|  W_{\{1,...,\beta+i\alpha\}})\right)\nonumber\\
					  & -\frac{2K-n+1}{n+1}(1+\epsilon F).
					\end{align}
					Then following the steps of proving Theorem \ref{th:2}, we have
						\begin{align}\label{eq:10ave}
		    R\geq \frac{2K-n+1}{n+1}-\frac{K(K+1)}{n(n+1)}\cdot \frac{M}{N}-\frac{2K-n+1}{n+1}(\epsilon+\frac{1}{F})
		\end{align}
					if the following inequality holds:
		\begin{align}
		    K\beta+\alpha\frac{(K-n)(K-n-1)}{2}\leq \frac{n(n+1)\alpha}{2}.\label{eq:th2condave}
		\end{align}
		Otherwise, we have
		\begin{align}
		    R\geq& \frac{2K-n+1}{n+1}-\frac{2K(K-n)}{n(n+1)}\cdot \frac{M}{N-\beta} \nonumber\\&-\frac{2K-n+1}{n+1}(\epsilon+\frac{1}{F}).
		\end{align}
		Similar to the proof of Lemma \ref{lemma:2}, we have proved that $R\geq R_\textup{u}(N,K,r)$ from the above bounds if  $r\in \left[\max\{\lceil\frac{K-3}{2}\rceil,1\}, K-1\right)$ for large $N$, large $F$, and small $\epsilon$. Consequently, we proved that $R^*_{\textup{ave}}=R_{\textup{u}}(N,K,r)$ if $r\leq 1$ or $r\geq \lceil\frac{K-3}{2}\rceil$ for large $N$. For systems with no more than $5$ users, this gives the exact characterization.

\section{Convexity of $R_{\textup{u}}(N,K,r)$ and $R_{\textup{u,ave}}(N,K,r)$}\label{app:convexity}

In this appendix, we prove the convexity of $R_{\textup{u}}(N,K,r)$ and $R_{\textup{u,ave}}(N,K,r)$ as functions of $r$, given parameters $N$ and $K$. We start by proving the convexity of $R_{\textup{u}}(N,K,r)$. 

Recall that for any non-integer $r$, the value of $R_{\textup{u}}(N,K,r)$ is defined by linear interpolation. Hence, it suffices to show that $R_{\textup{u}}(N,K,r)$ is convex on $r\in\{0,1,...,K\}$. Equivalently, we only need to prove 
\begin{align}\label{eq:127}
2R_{\textup{u}}(N,K,r)- R_{\textup{u}}(N,K,r-1)-R_{\textup{u}}(N,K,r+1)\leq 0
\end{align}
for any $r\in\{1,...,K-1\}$.

The proof is as follows. We first observer that $R_{\textup{u}}(N,K,r)$ can be written as
	\begin{align}
			    	R_{\textup{u}}(N,K,r)&=
				\frac{\binom{K}{r+1}-\binom{K-\min\{K,N\}}{r+1}}{\binom{K}{r}}\\
				&=\frac{\sum_{i=1}^{\min\{K,N\}} \binom{K-i}{r}}{\binom{K}{r}}\\
				&=\sum_{i=1}^{\min\{K,N\}} \frac{\binom{K-r}{i}}{\binom{K}{i}}.
			\end{align}
Consequently, the LHS of inequality (\ref{eq:127}) can be written as 
\begin{align}
2R_{\textup{u}}&(N,K,r)- R_{\textup{u}}(N,K,r-1)-R_{\textup{u}}(N,K,r+1)\nonumber\\
&= \sum_{i=1}^{\min\{K,N\}} \frac{2\binom{K-r}{i}-\binom{K-r-1}{i}-\binom{K-r+1}{i}}{\binom{K}{i}}\\
&=\sum_{i=1}^{\min\{K,N\}} \frac{\binom{K-r-1}{i-1}-\binom{K-r}{i-1}}{\binom{K}{i}}\\
&=\sum_{i=2}^{\min\{K,N\}} \frac{-\binom{K-r-1}{i-2}}{\binom{K}{i}}.
\end{align}
Since both  $\binom{K-r-1}{i-2}$ and $\binom{K}{i}$ are non-negative, we have proved inequality (\ref{eq:127}). This guarantees the convexity of $R_{\textup{u}}(N,K,r)$. 

Note that by substituting the variable $\min\{K,N\}$ in function $R_{\textup{u}}(N,K,r)$ by $N_{\textup{e}}(\boldsymbol{d})$, and taking expectation over a uniformly random demand $\boldsymbol{d}$, we exactly obtain  function $R_{\textup{u,ave}}(N,K,r)$. Consequently, by applying the same substitution in the above proof, we obtain a proof for the convexity of $R_{\textup{u,ave}}(N,K,r)$.
\bibliographystyle{ieeetr}
   \bibliography{UCache}

\section*{Biographies}

\begin{IEEEbiographynophoto}{Qian Yu} (S'16) is pursuing his Ph.D. degree in Electrical Engineering at University of Southern California (USC), Viterbi School of Engineering. He received his M.Eng. degree in Electrical Engineering and B.S. degree in EECS and Physics, both from Massachusetts Institute of Technology (MIT). His interests span information theory, distributed computing, and many other problems math-related.


Qian is a recipient of the Google PhD Fellowship in 2018, and received the Jack Keil Wolf ISIT Student Paper Award in 2017. He received the Annenberg Graduate Fellowship in 2015, and Honorable Mention in the William Lowell Putnam Mathematical Competition in 2013.

\end{IEEEbiographynophoto}

\begin{IEEEbiographynophoto}{Mohammad Ali Maddah-Ali} (S'03-M'08) received the B.Sc. degree from Isfahan University of Technology, and the M.A.Sc. degree from the University of Tehran, both in electrical engineering. From 2002 to 2007, he was with the Coding and Signal Transmission Laboratory (CST Lab), Department of Electrical and Computer Engineering, University of Waterloo, Canada, working toward the Ph.D. degree. From 2007 to 2008, he worked at the Wireless Technology Laboratories, Nortel Networks, Ottawa, ON, Canada.  From 2008 to 2010, he was a post-doctoral fellow in the Department of Electrical Engineering and Computer Sciences at the University of California at Berkeley. Then, he joined Bell Labs, Holmdel, NJ, as a communication research scientist. Recently, he started working at Sharif University of Technology, as a faculty member.
 
Dr. Maddah-Ali is a recipient of NSERC Postdoctoral Fellowship in 2007, a best paper award from IEEE International Conference on Communications (ICC) in 2014, the IEEE Communications Society and IEEE Information Theory Society Joint Paper Award in 2015, and the IEEE Information Theory Society Joint Paper Award in 2016. 
\end{IEEEbiographynophoto}

\begin{IEEEbiographynophoto}{A. Salman Avestimehr} (S'03-M'08-SM'17) is an Associate Professor at the Electrical Engineering Department of University of Southern California. He received his Ph.D. in 2008 and M.S. degree in 2005 in Electrical Engineering and Computer Science, both from the University of California, Berkeley. Prior to that, he obtained his B.S. in Electrical Engineering from Sharif University of Technology in 2003. His research interests include information theory, the theory of communications, and their applications to distributed computing and data analytics.

Dr. Avestimehr has received a number of awards, including the Communications Society and Information Theory Society Joint Paper Award, the Presidential Early Career Award for Scientists and Engineers (PECASE) for ``pushing the frontiers of information theory through its extension to complex wireless information networks'', the Young Investigator Program (YIP) award from the U. S. Air Force Office of Scientific Research, the National Science Foundation CAREER award, and the David J. Sakrison Memorial Prize. He is currently an Associate Editor for the IEEE Transactions on Information Theory.
\end{IEEEbiographynophoto}
   
\end{document}